\documentclass[a4paper,USenglish]{lipics-v2019}

\usepackage{cutwin}
\usepackage{adjustbox}
\usepackage{wrapfig}

\title{Fast Deterministic Algorithms for Highly-Dynamic Networks}

\titlerunning{Fast Deterministic Algorithms for Highly-Dynamic Networks} 

\author{Keren Censor-Hillel}
	{Technion---Israel Institute of Technology}
	{ckeren@cs.technion.ac.il}{}
	{This project has received funding from the European Union's Horizon
	2020 Research And Innovation Program under grant agreement no.755839.}
\author{Neta Dafni}
	{Technion---Israel Institute of Technology}
	{netad@cs.technion.ac.il}{}{}
\author{Victor I. Kolobov}
	{Technion---Israel Institute of Technology}
	{tkolobov@cs.technion.ac.il}{}{}
\author{Ami Paz}
	{Faculty of Computer Science, Universit\"at Wien}
	{ami.paz@univie.ac.at}{}
	{Supported by the Austrian Science Fund (FWF): P 33775-N, Fast Algorithms for a Reactive Network Layer.}
\author{Gregory Schwartzman}
	{Japan Advanced Institute of Science and Technology}
	{greg@jaist.ac.jp}{}
	{This work was supported by JSPS Kakenhi Grant Number JP19K20216 and~JP18H05291.}

\authorrunning{Censor-Hillel et al.} 

\Copyright{Keren Censor-Hillel, Neta Dafni, Victor I. Kolobov, Ami Paz, Gregory Schwartzman} 



\nolinenumbers 

\hideLIPIcs  

\usepackage{algorithm}
\usepackage{algorithmic}
\usepackage{amsfonts}
\usepackage{amsthm}
\usepackage{amssymb}
\usepackage{hyperref}
\usepackage{amsmath}
\usepackage{cite}

\usepackage{mathrsfs}
\usepackage{graphicx}
\usepackage{enumerate}
\usepackage{enumitem} 
\usepackage{tcolorbox}

\newcommand{\OPT}{\textnormal{OPT}}

\newenvironment{theorem-repeat}[1]{\begin{trivlist}
		\item[\hspace{\labelsep}{\bf\noindent Theorem \ref{#1} .}]\em }%
	{\end{trivlist}}
	\newenvironment{corollary-repeat}[1]{\begin{trivlist}
		\item[\hspace{\labelsep}{\bf\noindent Corollary \ref{#1} .}]\em }%
	{\end{trivlist}}

\newcommand{\qedsymb}{\qed}

\newcommand{\true}{\texttt{true}}
\newcommand{\false}{\texttt{false}}
\newcommand{\matched}{\texttt{matched-to-}}
\newcommand{\unmatched}{\texttt{unmatched}}
\DeclareMathOperator{\prep}{\Phi_\textnormal{prepare}}
\DeclareMathOperator{\fix}{\Phi_\textnormal{fix}}

\DeclareMathOperator{\incorrect}{\texttt{incorrect}}
\DeclareMathOperator{\changes}{\texttt{changes}}
\DeclareMathOperator{\epochs}{\texttt{epochs}}

\begin{document}

\maketitle

\begin{abstract}
This paper provides an algorithmic framework for obtaining fast distributed algorithms for a highly-dynamic setting, in which \emph{arbitrarily many} edge changes may occur in each round. Our algorithm significantly improves upon prior work in its combination of (1) having an $O(1)$ amortized time complexity, (2) using only $O(\log{n})$-bit messages, (3) not posing any restrictions on the dynamic behavior of the environment, (4) being deterministic, (5) having strong guarantees for intermediate solutions, and (6) being applicable for a wide family of tasks.

The tasks for which we deduce such an algorithm are maximal matching, $(degree+1)$-coloring, 2-approximation for minimum weight vertex cover, and maximal independent set (which is the most subtle case). For some of these tasks, node insertions can also be among the allowed topology changes, and for some of them also abrupt node deletions.
\end{abstract}

\section{Introduction}
\label{section:introduction}
\vspace{-0.1cm}
We present a family of deterministic distributed algorithms that rapidly fix solutions for fundamental tasks even in a highly-dynamic environment. Specifically, we provide algorithms for maximal matching, $(degree+1)$-coloring, 2-approximation for the minimum \emph{weighted} vertex cover (2-MWVC), and maximal independent set (MIS). 
We further show that for some of these tasks, fast fixing is also possible with node insertions and deletions. Here, we consider the severe case of \emph{abrupt} deletions, where  a deleted node does not have a chance to inform its neighbors about its upcoming departure from the system.

Our algorithms enjoy the combination of (1) having an $O(1)$ amortized time complexity, (2) using only $O(\log{n})$-bit messages, (3) not posing any restrictions on the dynamic behavior of the environment and in particular not requiring topology changes to be spaced in time, (4) being deterministic, (5) having strong guarantees for intermediate solutions, and (6) being applicable for a wide family of tasks. 
In recent years, there has been much progress on distributed dynamic algorithms, achieving different combinations of the above promises. Our algorithms significantly improve upon all prior work by that they guarantee the combination of all the above properties. We elaborate upon -- and compare to -- prior work in Section~\ref{subsec:related}.

We stress that as opposed to centralized dynamic data structures, not posing any restrictions on the dynamic behavior of the environment is vital in the distributed setting, as
the input graph is the communication graph itself. More concretely, in centralized dynamic data structures when multiple topology changes occur, we can simply handle them one by one. However, in our setting, nodes cannot communicate over a deleted edge, and so we cannot sequentially apply an independent update algorithm for each topology change ---  an edge deletion affects the communication already when it \emph{happens}, not only when it is \emph{handled}.

\subsection{Motivation}
\label{subsection:motivation}
\vspace{-0.1cm}
Each of the aforementioned problems is a locally-checkable labeling (LCL) problem. The notion of an LCL is a celebrated concept in distributed computing, first defined by Naor and Stockmeyer~\cite{NaorS95} in order to capture tasks in which nodes can efficiently detect inconsistencies, motivated by the unstable nature of distributed systems.
Since the publication of this pioneering work, the complexity of solving tasks that can be described as LCLs has been extensively studied in the distributed setting. We ask the following question, paraphrased in correspondence with the title of~\cite{NaorS95}:
%
\begin{center}
\emph{Question: What can be fixed locally?}
\end{center}
%
We begin by recalling the definition of LCLs of~\cite{NaorS95}, restricting our attention to LCLs with \emph{radius} $r=1$.
A \emph{centered star} is a pair $(H,s)$ where $H$ is a star graph and $s$ is its center.
An LCL $\mathcal{L}$ is a tuple $(\Sigma, \Gamma, \mathcal{C})$, where $\Sigma$ is a set of \emph{input labels}, $\Gamma$ is a set of \emph{output labels}, and $\mathcal{C}$ is a set of \emph{locally consistent labelings}. Each element of $\mathcal{C}$ is a centered star, with a label in~$\Sigma \times \Gamma$ for each of its nodes.\footnote{In the work of Naor and Stockmeyer~\cite{NaorS95} the set of labels $\Sigma$ has a fixed size, while here we omit this limitation in order to give more power to the labelings. However, algorithmically, we always keep the size of messages small even when labels are large, by sending only pieces of them.}

A labeling $\lambda: V \rightarrow \Sigma \times \Gamma$ is called $\mathcal{L}$-legal for a graph $G=(V,E)$, if for every $v\in V$, there exists a centered star $(H,s)$ in $\mathcal{C}$ with a label-pair at each node,
which is consistent with $\lambda$ in the following sense: there exists a mapping $\pi$ that maps the star centered at $v$ in $G$ into $(H,s)$, with $\pi(v)=s$, such that for every node $w$ in the star centered at $v$,
the label-pair given by $\lambda$ is the same as the label-pair of the node $\pi(w)$ in $(H,s)$.

As explained in~\cite{NaorS95}, the set $\mathcal{C}$ defines allowed labels for neighborhoods, as opposed to defining a set of forbidden ones. If the LCL has no inputs, then one can simply choose a default input label, i.e., $|\Sigma| = 1$. An algorithm that solves the problem defined by an LCL $\mathcal{L}$ is an algorithm whose output on a graph $G$ is an $\mathcal{L}$-legal labeling.

~\\\textbf{Not all LCLs are easily fixable:} The  following variant of the \emph{sinkless orientation} problem~\cite{BrandtFHKLRSU16} is an example of an LCL problem that is not easily fixable.
Each node has a label that corresponds to an orientation of its edges, such that labels at endpoints of an edge are consistent, and such that there is no node of degree greater than 1 that is a sink, i.e., has no outgoing edge.
It is easy to verify that every graph has a valid labeling%
\footnote{If $G$ is a tree, choose an arbitrary root and orient the edges away from the root.
Otherwise, choose a cycle in $G$ and orient its edges cyclically, then imagine contracting its nodes into a single super-node and orient edges towards this super-node along some spanning tree, and
orient other edges arbitrarily.}%
, and that this is an LCL.
To see that this LCL cannot be fixed within an amortized complexity of $O(1)$, consider a graph on $n$ nodes that evolves dynamically, creating two paths of roughly $n/2$ nodes each. Each path must be oriented consistently with a single sink in one of its endpoints. Inserting an edge between the sinks of the two paths forces the orientation of all of the edges in one of the sub-paths to flip, which takes $\Omega(n)$ rounds. Deleting this edge induces again two paths with a single sink each, and repeating the process of inserting an edge between the new sinks and deleting it causes a linear number of rounds that can be attributed to only two topology changes, which implies an amortized time of $\Omega(n)$. This holds even if topology changes do not happen concurrently, and even if the messages can be of arbitrarily large size.


\subsection{The challenges}
For any LCL problem we address, we assume that the system begins with a globally correct labeling, and thus what an algorithm needs to do as a consequence of topology changes is to have the affected nodes update their labels. Naturally, for some problems, the update procedure may also require that a node updates the labels of its neighbors (more precisely, this is accomplished via sending messages to its neighbors requiring them to update their labels).
For example, in a solution for maximal matching this might occur when an edge that is in the matching is deleted, and its endpoints need to be match themselves to other neighbors. 
At a first glance, this may sound as a simple and straightforward approach for fixing matchings and problems of local flavor. However, this approach turns out to be far from trivial, and below we describe multiple key challenges that we must overcome in order to implement it successfully.

~\\\textbf{(1) Defining \emph{fixing} and \emph{amortized complexity}:} We need to define what \emph{fixing the solution} means. We aim for our algorithm to work in a very harsh setting, in which it might be the case that there are so many topology changes that we never actually obtain a globally correct labeling, but still we maintain strong guarantees for intermediate labelings. 
Notice that this is in stark contrast to centralized dynamic data structures, which can always consider globally correct solutions since topology changes may be handled one-at-a-time because they only affect the input and not the computation itself. This is also the case for the majority of previous distributed algorithms: they are designed under the assumption that topology changes are spaced well enough in time so that it is possible to obtain a globally correct solution before the next topology change happens. 

~\\\textbf{(2) Coping with concurrent fixing with a timestamp mechanism:} Because we might need a node to change the labels of its neighbors and not only its own label in order to fix the solution, we make sure that concurrent fixing always happens for nodes that are not too close, and other nodes wait even if their labeled stars are not yet correct (e.g., to avoid two nodes $u,v,$ trying to get matched to the same node $w$ concurrently). To this end, our method is to assign a timestamp to each node involved in a change, and \emph{fix} a node only if its timestamp is a local minimum in some short-radius neighborhood, thus avoiding conflicting concurrent fixes. We call such a node \emph{active}.

~\\\textbf{(3) Detecting and aborting conflicting timestamps:} Such a timestamp mechanism alone is still insufficient: the uncontrolled number of topology changes may, for example, suddenly connect two nodes that were previously far enough so that they could become active simultaneously, but after concluding that they can both become active, an edge insertion now makes them part of the same short-radius neighborhood. We carefully take care of such cases where our timestamps have been cheated by the topology changes, by detecting such occurrences and \emph{aborting the fixing}, without harming the amortized complexity guarantees.

~\\\textbf{(4) Bounding the size of timestamps to cope with message size restrictions:} Finally, the restriction on the size of messages forbids unbounded timestamps, despite an unbounded number of rounds (e.g., times). To resolve this issue, we utilize ideas from the literature on shared memory algorithms, e.g.,~\cite{AttiyaDS89}, for deterministically hashing the timestamps into a small bounded domain so that the nodes can afford sending a hashed timestamp in a single small message, and we do so in a way that preserves the total order over timestamps. 

\subsection{Our contributions}
Our main contribution is thus deterministic dynamic distributed fixing algorithms for several fundamental problems. Our algorithms share a common approach, and only minor modifications that are specific to each labeling are required. 
In some cases we can also handle a node insertion/deletion, which is a-priori possibly harder to deal with, because it may affect more nodes while in the amortized analysis we count it as a single topology change.

The following theorem summarizes the end-results, which hold in a model with an unbounded number of topology changes that may occur concurrently, and when only a logarithmic number of bits can be sent in a message.
\begin{theorem}
\label{theorem:all}
	There is a deterministic dynamic distributed fixing algorithm for \textbf{$(degree+1)$-coloring} and for a \textbf{$2$-approximation of a minimum weight vertex cover}, which handles edge insertions/deletions and node insertions in $O(1)$~amortized rounds.

	There are deterministic dynamic distributed fixing algorithms for \textbf{maximal matching}, \textbf{$(\Delta+1)$-coloring} (where $\Delta$ is the maximum node degree) and \textbf{MIS}, which handle edge/node insertions/deletions in $O(1)$~amortized rounds.
\end{theorem}

Sections~\ref{section:mm} and~\ref{section:coloring} show our algorithm for maximal matching and $(degree+1)$-coloring, respectively. 
This is developed and modified in Section~\ref{sec:apx}, to present our 2-MWVC algorithm. We mention that the labeling for the solution of 2-MWVC that we maintain is not the na\"{i}ve one that only indicates which nodes are in the cover, but rather contains information about dual variables that correspond to edge weights, and allow the fast fixing.

Section~\ref{subsec:mis} gives our algorithm for MIS. In the MIS case, the restriction of message size imposes an additional, huge difficulty. The reason is that if an MIS node $v$ needs to leave the MIS because an edge is inserted between $v$ and some other MIS node $u$, then all other neighbors of $v$ who were previously not in the MIS are now possibly not covered by an MIS neighbor. Yet, they cannot all be moved into the MIS, as they may have an arbitrary topology among them. With unbounded messages this can be handled using very large neighborhood information but such an approach is ruled out by the the restriction of $O(\log n)$-bit messages.

Nevertheless, we prove that with some modifications to our algorithmic approach, we can also handle MIS without the need to inform nodes about entire neighborhoods. The road we take here is that instead of fixing its neighborhood, a node tells its neighbors that they should become active themselves in order to fix their labeled stars. On the surface, this would entail an unacceptable overhead for the amortized complexity that is proportional to the degree of the node. The crux in our algorithm and analysis is
in blaming previous topology changes for such a situation --- for every node $u$ in the neighborhood of $v$ which is only dominated by $v$, there is a previous topology change (namely, an insertion of an edge $\{u,w\}$, where $w$ may or may not be $v$) for which we did not need to fix the label of $w$. This potential function argument allows us to amortize the round complexity all the way down to $O(1)$, and the same technique is utilized to handle node insertions and deletions.

In Appendices~\ref{section:LFL general case}--\ref{app:alg}, we present a generalization of our algorithm, by defining a family of graph labelings, in the flavor of the LCL definition, which can all be fixed in constant amortized time.
This generalization is inherently intricate, and is added in order to assist a reader who may wonder about such a generalization, and to motivate our choice of presenting one algorithm for maximal matching and several modifications of it for the other problems, rather than a single unified algorithm.

\subsection{Related work}	
\label{subsec:related}
The end results of our work provide fast fixing for fundamental graph problems, whose static algorithmic complexity has been extensively studied in the distributed setting. A full overview of the known results merits an entire survey paper on its own (see, e.g.,~\cite{Suomela13,BarenboimEBOOK}). An additional line of beautiful work studies the landscape of distributed complexities of LCL problems, and the fundamental question of using randomness (see, e.g.,~\cite{BrandtHKLOPRSU17,GhaffariKM17,BalliuHKLOS18,BalliuBOS19,ChangP19,ChangKP19}).

For dynamic distributed computing, there is a rich history of research on the important paradigm of \emph{self-stabilization} (see, e.g., the book~\cite{Dolev2000}) and in particular on symmetry breaking (see, e.g., the survey~\cite{GuellatiK10}). Related notions of error confinement and fault-local mending have
been studied in \cite{AzarKP10,KuttenP99,KuttenP00}. Our model greatly differs from the above. There are many additional models of dynamic distributed computation (e.g.,~\cite{BonneC19, Kuhn2010}), which are very different from the one we consider in this paper.

Some of the oldest works in similar models to ours are~\cite{Italiano91, Elkin07}, who provide algorithms for distance-related tasks.
Constant-time algorithms were given 
in~\cite{KonigW13} for symmetry-breaking problems assuming unlimited bandwidth and a single topology change at a time.
The work of~\cite{Censor-HillelHK16},  provides a randomized algorithm that uses small messages to fix an MIS in $O(1)$-amortized update time for a non-adaptive oblivious adversary, still assuming a single change at a time.
The latter left as an open question the complexity of fixing an MIS in the sequential dynamic setting. This was picked up~in~\cite{AssadiOSS18, AssadiOSS19,DuZ2018,GuptaK18}, 
giving the first non-trivial sequential MIS algorithms, 
which were recently revised and improved~\cite{Behnezhad2019,Chechik2019}.
Specifically, the algorithm of~\cite{AssadiOSS18} achieves an $O(\min\{\Delta, m^{3/4}\})$ amortized message complexity and $O(1)$-amortized round complexity and adjustment complexity (the number of vertices that change their output after each update) for an adaptive non-oblivious adversary in the distributed setting.
However, they handle only a single change at a time, and sometimes need to know the number of edges, which is global knowledge that our work avoids assuming.
In fact, if one is happy with restricting the algorithm to work only in a model with a single topology change at a time, then sending timestamps is not required, so $O(1)$-bit messages suffice in our algorithm for MIS, resembling what~\cite{AssadiOSS18} obtains.
\cite{ParterPS16} provides a neat log-starization technique, which translates logarithmic static distributed algorithms into a dynamic setting such that their amortized time complexity becomes $O(\log^*{n})$. This assumes a single change at a time and large messages.
\cite{Solomon16} shows that maximal matching have $O(1)$ amortized complexity, even when counting messages and not only rounds, but assuming a single change at a time.

The ($\Delta+1$)-coloring algorithm of~\cite{BarenboimEG18} also implies fixing in a self-stabilizing manner --- after the topology stops changing, only $O(\Delta+\log^* n)$ rounds are required in order to obtain a valid coloring, where $\Delta$ bounds the degrees of all the nodes at all times.

Perhaps the setting most relevant to ours is the one studied 
in~\cite{BambergerKM18}, who also address a very similar highly-dynamic setting. They insightfully provide fast dynamic algorithms for a wide family of tasks, which can be decomposed into packing and covering problems, in the sense that a packing condition remains true when deleting edges and a covering condition remains true when inserting edges. For example, MIS is such a problem, with independence and domination being the packing and covering conditions, respectively.
An innovative contribution of their algorithms is providing guarantees also for intermediate states of the algorithm, that is, guarantees that hold even while the system is in the fixing process. They show that the packing property holds for the set of edges that are present throughout the last $T$ rounds, and that the covering property holds for the set of edges that are present in either of the last $T$ rounds, for $T=O(\log n)$. Moreover, their algorithms have correct solutions if a constant neighborhood of a node does not change for a logarithmic number of rounds. Our algorithm guarantees correctness of labeled stars for nodes for which any topology change touching their neighborhood has already been handled.
In comparison with their worst-case guarantee of $O(\log n)$ rounds for a correct solution, our algorithm only gives $O(n)$ rounds in the worst case. However, our amortized complexity is $O(1)$, our messages are of logarithmic size, and our algorithm is deterministic, while the above is randomized with messages that can be of polylogarithmic size.

In addition, a recent work~\cite{CKS} studies subgraph problems in the same model described in our paper.

A different definition of local fixability is given in~\cite[Appendix A]{BhattacharyaCHN18}, suitable for \emph{sequential} dynamic data structures, which requires a node to be able to fix the solution by changing only its own state. While this captures tasks such as coloring, and is helpful in the sequential setting for avoiding the need to update the state of all neighbors of a node, in the distributed setting we can settle for a less restrictive definition, as a single communication round suffices for updating states of neighbors, if needed. 
Indeed, our algorithmic framework captures a larger set of tasks: notably, we provide an algorithm for MIS, while~\cite{BhattacharyaCHN18} prove that it does not fall into their definition.
In addition, \cite[Section 7]{BhattacharyaCHN18} raise the question of fixing (in the sequential setting) problems that are in P-SLOCAL\footnote{Roughly speaking, SLOCAL($t$) is the class of problems that admit solutions by an algorithm that iterates over all the nodes of the graph, and assigns a solution to each node based on the structure of its $t$-neighborhood and solutions already assigned to nodes in this neighborhood. P-SLOCAL is the class SLOCAL(polylog $n$).}~\cite{GhaffariKM17}. Notably, this class contains approximation tasks, and indeed for some approximation ratios we can apply our framework. Indeed, our algorithm has the flavor of sequentially iterating over nodes and fixing the labels in their neighborhood, with the additional power of the distributed setting that allows it to work concurrently on nodes that are not too close. This also resembles the definition of orderless local algorithms given in~\cite{KawarabayashiS18}, although a formal definition for the case of fixing does not seem to be simpler than our framework.

\section{Model}
\label{section:model}
We assume a synchronous network that starts as an empty graph on $n$ nodes and evolves into the graph $G_i=(V_i,E_i)$ at the beginning of round $i$; in most of our algorithms, one can alternatively assume any graph as the initial graph, as long as the nodes start with a labeling that is globally consistent for the problem in hand. 
In some cases, we also allow node insertion or deletion, and then $n$ serves as a universal upper bound on the number of nodes in the system.
Each node is aware of its unique id, the edges it is a part of, its weight if there is one, and of $n$. In addition, the nodes have a common notion of time, so the execution is synchronous. 
New nodes do not know the global round number. (We mention that in our algorithms it is sufficient for each node to know the round number modulo $15n$, and a new node can easily obtain this value from its neighbors, so we implicitly assume all nodes have this knowledge.)

In each round, each node receives \emph{indications} about the topology changes that occurred to its incident edges. We stress that the indications are a posteriori, i.e., the nodes get them only after the changes occur, and thus cannot prepare to them in advance (these are called \emph{abrupt} changes). 
After receiving the indications and performing local computation, each node can send messages of $O(\log n)$ bits to each of its neighbors.

We work in a distributed setting where each node stores its own label. A distributed fixing algorithm should update the labels of the nodes in a way that corrects the labeled stars that become incorrect due to topology changes. Naturally, for a highly-dynamic setting, we do not require a global consistent labeling in scenarios in which the system is undergoing many topology changes.

We consider four classical graph problems.
In the \emph{maximal matching} problem, the nodes have to mark a set of edges
such that no two intersect, and such that no edge can be added to the set without violating this condition.
In \emph{minimum weight vertex cover} (MWVC), the nodes start with weights, and the goal is to choose a set of nodes that intersect all the edges, and have the minimum weight among all such sets; we will be interested in the $2$-approximation variant of the problem, where the nodes choose a set of weight at most twice the minimum.
Finally, in the \emph{maximal independent set} (MIS) problem, the nodes must mark a set of nodes such that no two adjacent nodes are chosen, and such that no node can be added to the set.

~\\\textbf{The complexity of distributed fixing algorithms:} 
When the labels of a star become inconsistent due to changes, a distributed fixing algorithm will perform a fixing process, which ends when the labels are consistent again, or when other changes occur in this star. 
The \emph{worst-case round complexity} of a distributed fixing algorithm is the maximum number of rounds such a fixing process may take.

In our algorithms, it could be that it takes a while to fix some star, but we can argue that this is because other stars are being fixed. We measure this progress with a definition of the amortized round complexity.

When studying \emph{centralized} algorithms for dynamic graphs, the amortized complexity measure is typically defined by an \emph{accounting method}, i.e., considering the time when the fixing process ends, and dividing the number of computation steps taken so far by the number of changes that occurred. The natural generalization of this definition to the distributed setting could be to take a  time when the graph labeling is globally correct, and divide the number of rounds occurred so far by the number of changes the network had undergone.
The first and most eminent problem in such a definition is that it requires a time when the \emph{global solution} is correct, which is something that we cannot demand in a highly-dynamic environment. 
The second problem with it is that the adversary can fool this complexity measure, by doing nothing for some arbitrary number of rounds in which the graph is correct, while the algorithm still gets charged for these rounds.

To overcome the above problems, we define the amortized round complexity as follows. 
Starting from round 0, in which the labeling is consistent for all stars, we consider the situation in each round $i$.
We denote by $\incorrect(i)$ the number of rounds until round $i$ in which there exists at least one inconsistent star. These are the computation rounds for which we charge the algorithm. Notice that we do not count only communication rounds in order to prevent an algorithm that cheats by doing nothing.\footnote{One could count also rounds in which the labeling is globally correct if the algorithm chooses to communicate in these rounds. Our algorithm never communicates in such rounds, so such a definition would not change our amortized complexity.}
We denote by $\changes(i)$ the number of changes which occurred until round $i$.
We say that an algorithm has an \emph{amortized round complexity $k$} if for every $i$ with $\changes(i)>0$, we have $\incorrect(i)/\changes(i)\leq k$.
This definition captures the \emph{rate} at which changes are handled, in a way that generalizes the sequential definition.

\subparagraph{Guarantees of our algorithm:}
Our algorithms have an $O(1)$ amortized fixing time, and in addition, they have additional desired progress properties.
First, our algorithms guarantee a worst-case complexity of $O(n)$, which implies that repeated changes far from a given star will not postpone it from being fixed for too long.
Moreover, if a labeled star is consistent and no topology change touches its neighborhood, then it remains consistent. Thus, our algorithm has strong guarantees also for intermediate solutions.

\section{An $O(1)$ amortized dynamic algorithm for maximal matching}
\label{section:mm}

The solution to the maximal matching problem at any given time is determined according to the labels of the nodes. A label of a node $v$ can be either $\unmatched$ or $\matched u$, indicating that $v$ is unmatched, or is matched to $u$, respectively.
Each node starts with the label $\unmatched$. Alternatively, one can assume any graph as the initial graph, as long as the nodes start with a legal maximal matching solution.
We prove the following.

\begin{theorem}
\label{theorem:mm}
There is a deterministic dynamic distributed fixing algorithm for maximal matching which handles edge insertions/deletions in $O(1)$ amortized rounds.
\end{theorem}

\begin{proof}
First, we assume that all nodes start with an initial globally consistent solution.

~\\\textbf{The setup:}
We denote $\gamma=5$.

Let $F_i$ be a set of edge changes (insertions/deletions) that occur in round $i \geq 0$ (for convenience, the first round is round 0). With each change in $F_i$, we associate two \emph{timestamps} such that a total order is induced over the timestamps as follows: for an edge $e=\{u,v\}$ in $F_i$, we associate the timestamp $ts=(i,u,v)$ with node $u$, and the timestamp $(i,v,u)$ with node $v$. Since $u$ and $v$ start round $i$ with an indication of $e$ being in $F_i$, both can deduce their timestamps at the beginning of round $i$. We say that a node $v$ is the \emph{owner} of the timestamps that are associated with it. In each round, a node only stores the largest timestamp that it owns, and omits the rest.

Notice that timestamps are of unbounded size, which renders them impossible to fit in a single message. To overcome this issue we borrow a technique of ~\cite{AttiyaDS89}, and we invoke a deterministic hash function $H$ over the timestamps, which reduces their size to $O(\log{n})$ bits, while retaining the total order over timestamps. The reason we can do this is that not every two timestamps can exist in the system concurrently. To this end, we define $h(i)=i \mod 3\gamma n$ and
$H(ts)=(h(i), u,v)$ for a timestamp $ts=(i,u,v)$, and we define an order $\prec_H$ over hashed timestamps as the lexicographic order of the 3-tuple, induced by the following order $\prec_h$ over values of $h$. We say that $h(i)\prec_h h(i')$ if and only if one of the following holds:
\begin{itemize}
\item $0 \leq h(i) < h(i') \leq 2\gamma n$, or
\item $\gamma n \leq h(i) < h(i') \leq 3\gamma n$, or
\item $2\gamma n \leq h(i) < 3\gamma n$ and $0 \leq h(i') < \gamma n$.
\end{itemize}
If two timestamps $ts=(i,v,u)$, $ts'=(i',v',u')$ are stored in two nodes $v,v'$ at two times $i,i'$, respectively, it holds that $ts<ts'$ (by the standard lexicographic order) if and only if  $H(ts) \prec_H H(ts')$. The reason that this holds despite the wrap-around of hashed timestamps in the third bullet above, is the following property that we will later prove: for every two such timestamps, it holds that $i'-i \leq \gamma n$. This implies $h(i) \prec_h h(i')$ whenever $i < i'$ despite the bounded range of the function $h$.

~\\\textbf{The algorithm:}
In the algorithm, time is chopped up into \emph{epochs}, each consisting of $\gamma$ consecutive rounds, in a non-overlapping manner. That is, epoch $j$ consists of rounds $i=\gamma j,\dots,\gamma (j+1)-1$.
For every epoch $j \geq 0$, we consider a set $D_j \subseteq V$ of \emph{dirty} nodes at the beginning of each epoch, where initially no node is dirty ($D_0=\emptyset$).
Some nodes in $D_j$ may become \emph{clean} by the end of the epoch, so at the end of the epoch the set of dirty nodes is denoted by $D'_j$, and it holds that $D'_j \subseteq D_j$. At the beginning of epoch $j+1$, all nodes that receive any indication of an edge in $F_i$ in the previous epoch are added to the set of dirty nodes, i.e., $D_{j+1} = D'_{j} \cup I_j$, where $I_j$ is the set of nodes that start round $i$ with any indication about $F_i$, for any $\gamma j \leq i \leq \gamma (j+1)-1$.

Intuitively, the algorithm changes the labels so that \emph{the labels at the \textbf{end} of the epoch are consistent with respect to the topology that was at the \textbf{beginning} of the epoch}, unless they are labels of dirty nodes or of neighbors of dirty nodes.

The algorithm works as follows. In epoch $j=0$, the nodes do not send any messages, but some of them enter $I_0$ (if they receive indications of edges in $F_i$, for $0 \leq i \leq \gamma-1$).

Denote by $N_{v}^i$ the neighborhood of $v$ in round $i$, denote by $L_{v}^i$ the label of $v$ at the beginning of round $i$, before the communication takes place, and denote by $\hat{L}^i_v$ the label at the end of the round. Unless stated otherwise, the node $v$ sets $\hat{L}^i_v \leftarrow L^i_v$ and $L^{i+1}_v \leftarrow \hat{L}^i_v$.
Now, consider an epoch $j>0$. On round $\gamma j$ every node $v\in D_j$ may locally change its label to indicate that it is unmatched, in case the edge between $v$ and its previously matched neighbor is deleted:
\begin{equation}
\label{eqn:LprepMM}
L_v^{\gamma j}=
\begin{cases}
\matched u, & \text{if } \hat{L}_v^{\gamma j-1}=\matched u \text{ and } u\in N^{\gamma j}_v\\
\unmatched, & \text{otherwise}
\end{cases}
\end{equation}
where $\hat{L}_v^{\gamma j-1}$ is the label that $v$ has at the \emph{end} of round $\gamma j-1=\gamma (j-1)+4$, which, as we describe below, may be different from its label $L_v^{\gamma j-1}$ at the beginning of the round.\footnote{We stress that one can describe our algorithm with labels that can only change at the beginning of a round, but we find the exposition clearer this way.}
Then, the node $v$ sends $L_{v}^{\gamma j}$ to its neighbors. These are the labels for the graph $G_{\gamma j}$ which the fixing addresses.

\newlength{\strutheight}
\settoheight{\strutheight}{\strut}
\begin{adjustbox}{valign=T,raise=\strutheight,minipage={\linewidth}}
\begin{wrapfigure}{r}{0pt}
	\includegraphics[trim=9cm 12cm 14cm 5.8cm, clip, scale=0.5]{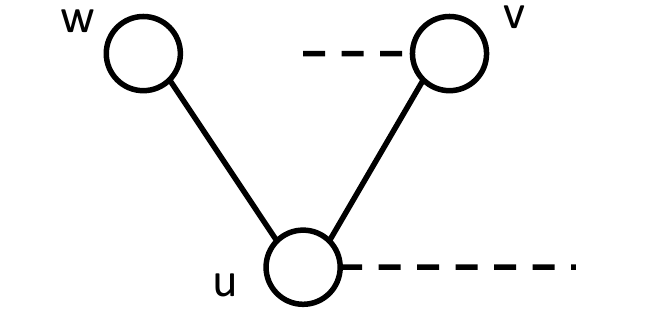}
\end{wrapfigure}
\strut{}
~~~ We stress that the new labels $L_v^{\gamma j}$ might not form consistent stars.
Instead, the nodes update $L_v^{\gamma j}$ and send it to all neighbors in order to maintain a common graph, with respect to which we show local consistency. As an example, consider a triangle $w,v,u$, undergoing the deletion of the edge $\{w,v\}$ and of another edge connecting $u$ with some other node (see Figure on the right). Suppose that $w$ immediately tries to fix the labels in its star, according to the fact that the edge $\{w,v\}$ does not exist, while $u$ is selected to fix its own star before $v$, without knowing of the deletion of the edge $\{w,v\}$. Both nodes then simultaneously try to change the label of $u$, and it could not be clear what $u$ should do, and which neighborhood of $u$ will be corrected. 
\vspace{1mm}
\end{adjustbox}

We continue describing the algorithm. On rounds $\gamma j+1$ to $\gamma j+3$ the nodes propagate the hashed timestamps owned by dirty nodes. That is, on round $\gamma j+1$, each node in $D_j$ broadcasts its hashed timestamp, and on the following two rounds all nodes broadcast the smallest hashed timestamp that they see (with respect to the order $\prec_H$). Every node $v$ in $D_j$ which does not receive a hashed timestamp that is smaller than its own becomes \emph{active}.

On the last round of the epoch, $\gamma j+4$, every active node $v$ computes the following candidate for a new label, denoting by $N^{\gamma j}_v=\{u_1,\ldots,u_d\}$ the neighborhood it had at round~$\gamma j$.
\begin{equation}
\label{eqn:LfixMM}
\ell_v =
\begin{cases}
\matched u_i, & \text{if } {L}_v^{\gamma j}=\matched u_i\\
\unmatched, & \text{if } {L}_v^{\gamma j}=\unmatched \text{ and } \text{for every } 1\leq i\leq d, L^{\gamma j}_{u_i} \neq \unmatched\\
\matched u_i, & \text{if } {L}_v^{\gamma j}=\unmatched \text{ and } 1\leq i\leq d \text{ is the smallest index}\\ &\text{for which }  L^{\gamma j}_{u_i} = \unmatched
\end{cases}
\end{equation}

Notice that $v$ has the required information to compute the above, even if additional topology changes occur during the rounds in which timestamps are propagated. Yet, we need to cope with the fact that topology changes may occur also throughout the current epoch and, for example, make active nodes suddenly become too close. For this, we denote by $T_j\subseteq I_j$ the set of \emph{tainted} nodes who received an indication of a topological change for at least one of their edges during the epoch $j$.

Now, only an active node $v$ which is not in $T_j$ sets $\hat{L}_v^{\gamma j+4} \leftarrow \ell_v$ and sends this new label to each neighbor $u$. 
Otherwise, an active node $v$ that is tainted (i.e., is in $T_j$) aborts and remains dirty for the next epoch.
Of course, if nodes $u$ and $v$ are neighbors at the beginning of an epoch but not when $v$ sends the computed label, then $u$ does not receive this information.

Finally, every active node $v\notin T_j$, if $\hat{L}^{\gamma j +4}_v = \matched u$ then $u$ updates $\hat{L}^{\gamma j +4}_u = \matched v$ (note that such $u$ has the required information since it receives $\ell_v$, as otherwise, if by the time that $\ell_v$ is computed it holds that $u$ and $v$ are no longer neighbors, then $v$ must be tainted).
At the end of round $\gamma j+4=\gamma (j+1)-1$, node $v$ becomes \emph{inactive} (even if it aborts) and is not included in $D'_j$, i.e., we initialize $D'_j=D_j\setminus \{v\mid v \text{ is active in epoch } j\}$ at the beginning of epoch $j$. Note that if $v$ is active but aborts then it is in $I_j$ and thus in $D_{j+1}$.

~\\\textbf{Round complexity:} We now prove that the algorithm has an amortized round complexity of $O(1)$, by proving $\incorrect(i)\leq 2\gamma\cdot\changes(i)$ for all $i$.
First, note that the algorithm communicates in each round where the graph is incorrect, and these communication rounds can be split into epochs, implying $\incorrect(i)\leq\gamma\cdot\epochs(i)$, where $\epochs(i)$ denotes the number of epochs of computation done by the algorithm until round $i$ (if round $i$ is the middle of an epoch then it does not affect the asymptotic behavior, so we can safely  ignore this partial epoch).
On the other hand, the node with minimal timestamp at the beginning of the $j$-th epoch becomes active during the epoch, and its timestamp is handled --- even if it becomes tainted by a change, the old timestamp is replaced by the new one. So, in each epoch at least one timestamp disappears from the system.
Now, since each topology change creates at most two timestamps, we have that the number of timestamps created until round $i$ is at most $2\cdot\changes(i)$, implying $\epochs(i)\leq 2\cdot\changes(i)$, and the claim follows.

Finally, we show that the timestamps can be represented by $O(\log n)$ bits. 
First, we claim that for every two timestamps $ts=(i,v,u)$ and $ts'=(i',v',u')$ such that $ts<ts'$, that are simultaneously owned by nodes at a given time, it holds that $i'-i \leq \gamma n$. Assume otherwise, and consider the first time when this condition is violated by a timestamp $ts'$, with respect to a previous timestamps $ts<ts'$. This means that the owner $v$ of $ts$ does not become active for more than $n$ epochs. Since up to this point in time there were no violations, in each epoch at least one timestamp was handled, and this was done in the desired ordered, i.e., all these labels where smaller than $ts$. So, $v$ not becoming active for more than $n$ epochs can only happen if at round $i$ there were more than $n$ timestamps which were then not yet handled, stored in various nodes. But there are at most $n$ nodes and each one stores at most one timestamp so the above is impossible. Since $i'-i \leq \gamma n$, we have that $H(ts) \prec_H H(ts')$, because $h(i) \prec_h h(i')$, as argued earlier.

Using the above we can also see that the worst case running time of our algorithm is $O(n)$. To see this, fix some node $v$ with an inconsistent star which does not experience topology changes touching its 1-hop neighborhood for $(\gamma+1) n$ rounds. This guarantees that its timestamp does not change throughout these rounds, and after $\gamma n$ rounds its timestamp must become a local minima. In the following epoch, if no changes occur within its 1-hop neighborhood then its star becomes consistent, which matches the definition of having a worst-case complexity of $O(n)$. Further, once a node $v$ successfully invokes a fixing of its star, the star remains consistently labeled as long as no topology changes touch the 1-hop neighborhood of $v$, thus we obtain strong guarantees for intermediate solutions.

\subparagraph{Correctness:}
For correctness we claim the following invariant holds at the end of round $i=\gamma j+4=\gamma(j+1)-1$:
For every two nodes $u,v$ that are clean at the end of the epoch and for which $\{u,v\}$ is an edge in $G_{\gamma j}$, it holds that (1) at least one of 
$\hat{L}_{u}^{\gamma j+4}$ and $\hat{L}_{v}^{\gamma j+4}$ is not $\unmatched$ and (2) if $\hat{L}_{u}^{\gamma j+4}=\matched v$ then $\hat{L}_{v}^{\gamma j+4}=\matched u$.

We prove the above by induction on the epochs. The base case holds trivially as during the first epoch the labels do not change, and we assume that the nodes start with a legal maximal matching for the initial graph. Now, assume the above invariants hold for epoch~$j-1$. 

For every two nodes $u,v$ that are clean at the end of the epoch and for which $\{u,v\}$ is an edge in $G_{\gamma j}$, if their labels do not change during the epoch, then the invariant follows from the induction hypothesis.

If only one of their labels changes, say that of $v$, then either $v$ is active and not tainted or there is a (single) neighbor $w$ of $v$ which is active and not tainted and makes $v$ change its label. In the former case, since the label $\ell_v$ of $v$ changes compared to $L_v^{\gamma j}$, it does not remain $\unmatched$ and does not remain $\matched x$ for some node $x$. So the new label $\ell_v$ must be $\matched y$, for some node $y$. Since the label of $u$ does not change, we have that $u\neq y$, and so if the label of $u$ is not $\unmatched$ then it cannot be $\matched v$ (as otherwise $L_v^{\gamma j}$ would be $\matched u$ and so $\ell_v$ would also be $\matched u$, thus did not change). In the latter case, if $v$ changes its label because of the new label $\ell_w$ that is sent to it by a neighbor $w$, then $\ell_w = \matched v$ and hence the new label of $v$ is set to $\matched w$.

Finally, if both of their labels change, then without loss of generality $v$ is active and not tainted and computes $\ell_v = \matched u$, making $u$ update its label to $\matched v$.
The crucial thing to notice here is that it cannot be the case that a node $w_v$ changes the label of $v$ and a different node $w_u$ changes the label of $u$ at the same time, because this implies that the distance between $w_v$ and $w_u$ is at most 3, in which case either at least one of them aborts due to an edge insertion, or the edge $\{u,v\}$ is inserted (maybe immediately after being deleted), but then $v$ and $u$ are not clean.

Since the invariant holds, we conclude that whenever $D_j=\emptyset$, it holds that the labeling is that of a maximal matching for $G_{\gamma j}$. Further, what the invariant implies is that some correctness condition holds even for intermediate rounds: at the end of every epoch $j$, the entire subgraph induced by set of nodes that are clean and have all of their neighborhood clean has labelings that are locally consistent.
\end{proof}

For node insertions and deletions, a direct application of the algorithm of Theorem~\ref{theorem:mm} increases the amortized complexity if all neighbors of a changed node (inserted or deleted) become dirty and $O(\Delta)$ timestamps are associated with this topology change. However, notice that when an edge is inserted, it suffices that \emph{only one} of its endpoints becomes dirty in the algorithm and gets matched to the other endpoint if needed. Hence, if a node is inserted, it suffices that the inserted node becomes dirty, and we do not need all of its neighbors to become so. An only slightly more subtle rule for deciding which nodes become dirty upon a node deletion gives the following.

\begin{theorem}
\label{theorem:mm-node}
There is a deterministic dynamic distributed fixing algorithm for maximal matching which handles edge/node insertions/deletions in $O(1)$ amortized rounds.
\end{theorem}

\begin{proof}
We modify the algorithm of Theorem~\ref{theorem:mm} as follows. Upon an insertion of a node $v$, the node $v$ becomes dirty. Upon a deletion of a node $v$ with neighbors $\{u_1,\ldots,u_d\}$, only the node $u_i$, for $1\leq i \leq d$, that is matched to $v$ (if there exists such a node) becomes dirty. 

The $O(1)$ amortized round complexity remains, as every topology change induces at most two new timestamps. Correctness still holds because it is not affected by a node insertion, which can be viewed as multiple edge insertions (in terms of correctness, but without paying this cost for the amortized time complexity), and it is not affected by a node deletion because for any other node $u_j \in \{u_1,\ldots,u_d\}$ such that $j \neq i$ it holds that the deletion of $v$ does not influence its local consistency.
\end{proof}

\section{An $O(1)$ amortized dynamic algorithm for coloring}
\label{section:coloring}
In the \emph{$c$-coloring} problem, each node must choose a color in $\{1,\ldots,c\}$, such that no two adjacent nodes have the same color.
In order to maintain a coloring in a dynamic graph, we use the same approach used for maintaining a maximal matching in the proof of Theorem~\ref{theorem:mm}.
Naturally, we use different labels for indicating a valid coloring, and so we modify the way in which they are updated during the algorithm. Yet, the overall structure of the algorithm remains the same. 

\begin{theorem}
	\label{theorem:deg-color}
	There is a deterministic dynamic distributed fixing algorithm for $(degree+1)$-coloring, which handles edge insertions/deletions in $O(1)$ amortized rounds.
\end{theorem}

\begin{proof}
	We describe the modifications that we make in the algorithm for maintaining a maximal matching from the proof of Theorem~\ref{theorem:mm}. For $(degree+1)$-coloring, the label of each node is its color. We use the same label notations $L_v^i$ as in our algorithm for maximal matching in the proof of Theorem~\ref{theorem:mm}, and we modify the way in which $L^{\gamma j}$ and $\ell$ are assigned new values in the assignments (\ref{eqn:LprepMM}) and (\ref{eqn:LfixMM}) of the algorithm, to correspond to a solution for $(degree+1)$-coloring rather than for maximal matching.
	
	The label $L_v^{\gamma j}$ stays the same as in the previous round (i.e., remains equal to $\hat{L}_v^{\gamma (j-1) + 4}$), unless the number of neighbors of $v$ decreases in $N_v^{\gamma j}$ compared with $N_v^{\gamma (j-1)}$, in which case $L_v^{\gamma j}$ is assigned to be $|N^{\gamma j}_v|$, where $N^{\gamma j}_v$ is the current star centered at $v$. Note that this gives a color that is in the correct range, but this may not be a valid coloring, as two neighbors may be assigned the same color if their neighborhoods have the same size.
	
	For the assignment of $\ell_v$ for every active node $v$, we set $\ell_v$ to be the minimal color not in $\{L_u^{\gamma j} ~\mid~ u \neq v, u \in N^{\gamma j}_v\}$. Notice that this means that the color of $v$ is valid. If $v$ is active and not in $T_j$ then $v$ sets $\hat{L}_v^{\gamma j +4}=\ell_v$ and sends $\ell_v$ to all of its neighbors. However, the neighbors do not have to change their own labels as a result (as opposed to the maximal matching algorithm, in which the label of $v$ may indicate to a neighbor $u$ that it is newly matched to $v$, in which case $u$ also updates its label $\hat{L}_{u}^{\gamma j +4}$).
	
	This guarantees the correctness of the color of $v$ w.r.t. its neighbors for every clean node $v$ for which all neighbors are clean . In particular, if all nodes are clean then the labels induce a proper $(degree+1)$-coloring. 
\end{proof}

The above algorithm directly applies also for node insertions, because, as in the maximal matching case, when a node $v$ is inserted it is sufficient to mark $v$ as dirty rather than also marking all of its neighbors. This holds since connecting the neighbors to $v$ does not invalid their colors, except maybe w.r.t. to color of $v$ itself, which will be fixed once $v$ becomes active. Specifically, their degrees do not decrease, so all their colors are in the relevant palates.

\begin{corollary}
	\label{corollary:deg-color}
	There is a deterministic dynamic distributed fixing algorithm for $(degree+1)$-coloring, which handles edge insertions/deletions and node insertions in $O(1)$ amortized rounds.
\end{corollary}

For node deletions, however, our algorithm incurs a large overhead, since the degrees of all the neighbors of a deleted node are decreased, which might require all of them to update their colors in the worst case. However, a slight relaxation into $(\Delta+1)$-coloring allows us to handle node deletions, since in this case the removal of any edge does not invalid the color of its endpoints, and thus no node needs to be marked as dirty upon node deletion.
\begin{theorem}
	\label{theorem:delta-color}
	There is a deterministic dynamic distributed fixing algorithm for $(\Delta+1)$-coloring which handles edge/node insertions/deletions in $O(1)$ amortized rounds.
\end{theorem}

\section{An $O(1)$ amortized dynamic algorithm for 2-MWVC}
\label{sec:apx}

In unweighted graphs, taking the endpoints of a maximal matching immediately gives a 2-approximation of the minimum vertex cover, so our algorithm for maximal matching also allows us to maintain a 2-approximate of minimum vertex cover in unweighted graphs.
Clearly, taking the endpoints of a maximal matching is insufficient for obtaining a 2-approximation of the minimum \emph{weight} vertex cover (henceforth, MWVC). Instead, we employ a primal-dual approach: 
the nodes maintain an edge-weight function which represents the dual of the MWVC problem, and the vertex cover will be composed of all nodes whose constraints are tight. By dual-fitting, this gives a 2-approximation of the MWVC, as desired.

We start with defining the general framework. 
Let $G=(V,w,E)$ be a weighted graph, where $w: V \rightarrow \mathbb{R}^{+}$ is a vertex-weight function.
We define the dual weights by an edge-weight function $\delta: E \rightarrow \mathbb{R}^{+}$. We say that $\delta$ is \emph{$G$-valid} if for every $v \in V$, $\sum_{e: v \in e}{\delta(e)} \leq w(v)$, i.e., the sum of weights of edges that touch a vertex is at most the weight of that node in $G$.

We define the set $S_{\delta}=\{v \in V \mid w(v)\leq\sum_{e: v \in e}{\delta(e)}\}$ of nodes with tight constraints (and when $\delta$ is $G$-valid, we can require $w(v)=\sum_{e: v \in e}{\delta(e)}$ instead). The following theorem states that if $\delta$ is $G$-valid and $S_{\delta}$ is a vertex cover, then it is a $2$-approximation for MWVC.
This theorem can be proved either using the primal-dual framework~\cite{BarYehudaE81}, or the local-ratio technique~\cite{BarYehuda00}, and we also give a simple proof of it in Appendix~\ref{sec:pf-lr}.
\newcommand{\TheoremLR}{
	Let $\OPT$ be the minimal weight of a vertex cover of $G$.
	If $\delta$ is a $G$-valid function, then $\sum_{v \in S_{\delta}}{w(v)} \leq 2\OPT$.
	In particular, if $S_{\delta}$ is a vertex cover then it is a $2$-approximation for MWVC for $G$.
}
\begin{theorem}
	\label{thm:lr}
	\TheoremLR
\end{theorem}

Using the above we continue to proving our main result. We use the same mechanism as we used in our algorithm for maximal matching in the proof of Theorem~\ref{theorem:mm}, with slight changes which we describe below.

\begin{lemma}
	\label{lemma:apx}
	There is a deterministic dynamic distributed fixing algorithm for 2-approximate MWVC, which handles edge insertions/deletions in $O(1)$ amortized rounds.	
\end{lemma}

\begin{proof}
	We assume here a graph $G=(V,w,E)$ with node weight function $w:V\to\{1,\ldots,W\}$, where $W$ is polynomial in $n$. 
	
	We use the same label notations $L^i_v$ for the label of node $v$ in round $i$. 
	Here the label of every node $v$ is an array of size $n$, where entries correspond to vertices in the graph. The entry $L_v^{i}[u]$ corresponds to the dual weight  $\delta(\{u,v\})$ of the edge $\{u,v\}$, if this edge exists. We call every such entry an \emph{edge excerpt} of the label.
	The special case of entry $L_v^{i}[v]$ corresponds to the remaining weight of the node $v$, which is $w(v)-\sum_{e: v \in e}{\delta(e)}$ . We call this entry the \emph{node excerpt} of the label. Initially, for every node $v$ the node excerpt is $L_v^{0}[v]=w(v)$ and all edge excerpts are 0.
	We use the same approach as in our algorithm for maximal matching in the proof of Theorem~\ref{theorem:mm}, and we modify the way in which $L^{\gamma j}_v$ and $\ell_{v}$ are assigned new values in the assignments (\ref{eqn:LprepMM}) and (\ref{eqn:LfixMM}) of the algorithm, to correspond to a solution for 2-MWVC rather than for maximal matching.
	
	Notice that now the labels are arrays, which results in large labels which cannot be sent using $O(\log n)$ bits of communication. 
	To overcome this, we modify the algorithm such that a node $v$ sends to each of its neighbors $u$ only $v$'s node excerpt $L_v^{i}[v]$, and the edge excerpt of their common edge $L_v^{i}[u]$.
	
	This is the only change we make with respect to the communication of the algorithm. Thus, it follows that the amortized running time remains $O(1)$, and it only remains to describe how we update $L_v^{\gamma j}$ at the beginning of every epoch (the analogue to assignment (\ref{eqn:LprepMM}) in the proof of Theorem~\ref{theorem:mm}), and how we update $\ell_v$ at the end of each epoch (the analogue to assignment (\ref{eqn:LfixMM}) in the proof of Theorem~\ref{theorem:mm}).
	
	Our desired properties from $L_v^{\gamma j}$ are: (P1) $\sum_{u \in V}{L_v^{\gamma j}[u]}=w(v)$, and (P2) $0 \leq L_v^{\gamma j}[v]\leq w(v)$. That is, $L_v^{\gamma j}[v]$ is indeed the remaining weight of $v$ and it is non-negative, and does not exceed its initial weight. We define $L_v^{\gamma j}$ as follows. For every $u \not\in N_v^{\gamma j -1} \oplus N_v^{\gamma j}$, $u\neq v$, we set $L_v^{\gamma j}[u]=\hat{L}_v^{\gamma j-1}[u]$, that is, these excerpts do not change. For every $u \in N_v^{\gamma j -1} \oplus N_v^{\gamma j}$ we set $L_v^{\gamma j}[u]=0$. Finally, we set $L_v^{\gamma j}[v]=\hat{L}_v^{\gamma j-1}+\sum_{u \in N_v^{\gamma j -1} \setminus N_v^{\gamma j}} \hat{L}_v^{\gamma j-1}[u]$, which means that the weight reductions over the removed edges are revoked. By a simple calculation, these updates guarantee that properties (P1) and (P2) hold.
	
	When computing $\ell_v$, we require: (P3) either $\ell_{v}[v]=0$ or for all $u\in N_v^{\gamma j}$ it holds that $\ell_u[u]=0$. For a labeling and a graph $G$, if these conditions hold for every $v\in V$ at round $i$, then it is immediate to see that the function $\delta: E \rightarrow \mathbb{R}^{+}$ defined by $\delta(\{u,v\})=L_{v}^i[u]=L_{u}^i[v]$ is $G$-valid, and the set $S=\{v \in V\mid L_{v}^i[v]=0\}$ is a vertex cover, which by Theorem~\ref{thm:lr} implies that $S$ is a 2-approximation for MWVC.
	
	Therefore, we define the assignment to $\ell_{v}$ as follows (the analogue to assignment~(\ref{eqn:LfixMM}) in the proof of Theorem~\ref{theorem:mm}). Denote the neighbors of $v$ at the corresponding round as $N_v^{\gamma j} = \{u_1, \dots, u_d \}$. For every $1\leq k \leq d$, we set $\ell_v[u_k] = L_{v}^{\gamma j}[u_k]-t_k$, where $t_k=\min\{L_{u_k}^{\gamma j}[u_k],L_{v}^{\gamma j}[v]-\sum_{1\leq p \leq k-1}{t_p}\}$. In addition, we set $\ell_{v}[v] = L_{v}^{\gamma j}[v]-\sum_{1\leq k \leq d}{t_k}$. Recall that if $v$ is active and not in $T_j$ then $v$ sets $\hat{L}_{v}^{\gamma j+4}$ to be $\ell_v$, and sends $\ell_v[u_k]$ to each $u_k$.  After receiving $\ell_v[u_k]$, each node $u_k$ sets $\hat{L}_{u_k}^{\gamma j+4}[v] = \ell_v[u_k]$, and then computes $t_k=\ell_v[u_k]-L_{u_k}^{\gamma j}[u_k]$ and sets $\hat{L}_{u_k}^{\gamma j+4}[u_k] = L_{u_k}^{\gamma j}[u_k]+t_k$. Simple calculations show that this guarantees that property~(P3) holds.
\end{proof}

For node insertions, we note that marking only the inserted node as dirty is sufficient, assuming the nodes of an empty graph start with labels for which each node excerpt is $w(v)$ and all edge excerpts are 0. This is because when the inserted node becomes active and gets fixed, it has a consistent star and sends updates to the relevant excerpts of neighboring nodes, in a way that maintains consistent stars for both the inserted node and its neighborhood. We thus have:

\begin{corollary}
	\label{cor:apx}
	There is a deterministic dynamic distributed fixing algorithm for a 2-approximation of a minimum weight vertex cover, which handles edge insertions/deletions and node insertions in $O(1)$ amortized rounds.
\end{corollary}

The question of handling node deletions for 2-MWVC remains open, as a direct marking of all of its neighbors as dirty incurs an amortized overhead that is proportional to the number of neighbors.

\section{An $O(1)$ amortized dynamic algorithm for MIS}
\label{subsec:mis}

One can use a similar approach in order to obtain an MIS algorithm. However, when a node $v$ needs to be removed from the MIS due to an edge insertion, a neighbor $u$ of $v$ may need to join the MIS if none its other neighbors are in the MIS. 
One way to do this is to mark all of the neighbors of $v$ as dirty, but this violates the amortized time complexity as a single topology change may incur too many dirty nodes. Another option is to have $v$'s label include neighborhood information such that upon receiving this label, its neighbors know which of them should be moved into the MIS.
This would give a simple $O(1)$ amortized rounds algorithm for MIS, but only if messages are allowed to be large. 
Instead, we present a labeling that does not contain all the neighborhood information and uses only \emph{small} labels.

To handle the new subset of neighbors that needs to be added to the MIS in the aforementioned example, our approach is to have the active node simply indicate to all of its neighbors that they cannot remain clean and must check for themselves whether they need to change their labels. Of course, such a single topology change may now incur a number of dirty nodes that is the degree of this endpoint, which may be linear in $n$.

Yet, we make a crucial observation here: any node that becomes dirty in this manner, can be blamed on a previous topology change in which only one node becomes dirty. This implies a potential function for the budget of dirty nodes, to which we add 2 units for every topology change, and charge either $0$, $1$, $2$, or $d$ (current node degree) units for each invocation of the fixing function, in a manner that preserves the potential non-negative at all times. Note that due to the potential function argument, here we must start with an empty graph for the amortization to work, unlike previous problems, where we could start with any graph as long as the nodes have labels that indicate a valid solution. Roughly speaking, we rely on the fact that since all nodes that are in the graph start as MIS nodes because there are no edges, then a node switches from being an MIS node to being a non-MIS node only upon an insertion of an edge between two MIS nodes, and the other endpoint of the inserted edge safely remains in the MIS. Note that we impose the rule that an inserted node never makes a node switch from being an MIS node to being a non-MIS node, since the inserted node chooses to become an MIS node only if all of its neighbors are already non-MIS nodes.

\begin{theorem}
\label{theorem:mis}
There is a deterministic dynamic distributed fixing algorithm for MIS, which handles edge/node insertions/deletions in $O(1)$ amortized rounds.
\end{theorem}

\begin{proof}
We consider labels which are in $\{\true,\false\}$ and maintain that the set of nodes with the label $\true$ form an MIS. We start with an empty graph and all labels are $\true$.
We first define the assignments of $L_v^i$ and $\ell_v$ as in assignments (\ref{eqn:LprepMM}) and (\ref{eqn:LfixMM}) in the algorithm for maximal matching in the proof of Theorem~\ref{theorem:mm}. Then, we explain how we modify the algorithm further in order to avoid large messages with neighborhood information.

First, the label for $L_v^{\gamma j}$ does not change from the previous round, i.e., assignment (\ref{eqn:LprepMM}) is $L_v^{\gamma j} = \hat{L}_v^{\gamma (j-1)+4}$.
For assignment (\ref{eqn:LfixMM}) we set $\ell_v$ to be $\false$ if $L_{u_i}^{\gamma j}=\true$ for some $u_i\in N_v^{\gamma j}$ and otherwise we set $\ell_v$ to be $\true$.
If $v$ is active and not in $T_j$ then it sets $\hat{L}_v^{\gamma j + 4}$ to be $\ell_v$ and sends this label to all of its neighbors.
Notice that this is insufficient for arguing that the labels at the end of the epoch form an MIS if all nodes are clean, for the same reason as in the tricky example above: if $v$ leaves the MIS due to an edge insertion, its neighbors do not have enough information to decide which of them joins the MIS.
To overcome this challenge, we consider an algorithm similar to the one of Theorem~\ref{theorem:mm}, with the following modifications.

\begin{enumerate}
\item[\textbf{(1)}] When an edge $e=\{v,u\}$ is deleted, if the labels of both $u$ and $v$ are $\false$ then neither of them becomes dirty, and if only one of them is $\false$ then only this node becomes dirty.
\item[\textbf{(2)}] When an edge $e=\{v,u\}$ is inserted, if at least one of the labels of $u$ and $v$ is $\false$ then neither of them becomes dirty, and if both are $\true$ then only the node with smaller $ID$ becomes dirty.
\item[\textbf{(3)}] When a node $v$ is inserted then only $v$ becomes dirty. 
\item[\textbf{(4)}] When a node $v$ is deleted then a neighbor $z$ becomes dirty only if its label is $\false$ and it has no neighbor with a label $\true$.
\end{enumerate}

In order for a node $v$ to indicate that new labels may be needed for its neighbors, we add the following item:
\begin{enumerate}
\item[\textbf{(5)}] When an active node $v$ changes its label to $\false$, all of its neighbors marked $\false$ that do not have a neighbor marked $\true$ become dirty.
\end{enumerate}

As we prove in what follows, this allows the correct fixing process that we aim for, but this has the cost of having too many nodes become dirty. However, the crucial point here is that not all nodes that become dirty in items~(4) and~(5) will actually utilize their timestamp --- some will drop their timestamp before competing for becoming active, and hence we will not need to account for fixing them. That is, we add the following item:
\begin{enumerate}
\item[\textbf{(6)}] When an active node $v$ changes its label to $\true$, all of its dirty neighbors become clean.
\end{enumerate}

\subparagraph{Correctness:} The correctness follows the exact line of proof of the algorithm in Theorem~\ref{theorem:mm}, with the modification that making some neighbors dirty in item~(5) compensates for not being able to assign them directly with good new labels.
That is, at the end of the epoch, we still have the following guarantee: if all nodes are clean, then their labels induce an MIS;
otherwise, for every two clean neighbors, either exactly one of them is in the MIS, or both have a neighbor in the MIS.

\subparagraph{Amortized round complexity:} 
The proof follows the same lines as the previous complexity proofs, with the addition of a potential function argument. This is used to prove that the cumulative number of epochs in which any node becomes active is at most twice the number of topology changes. This proves our claim of an amortized $O(1)$ round complexity.

First, as in the former algorithms, we note that $\incorrect(i)\leq\gamma\cdot\epochs(i)$, and at each epoch at least one timestamp is handled. Thus, we only need to upper bound the number of timestamps created by round $i$ as a function of $\changes(i)$. However, here we need to be much more careful and we can not simply account each change for two timestamps, as some changes create much more timestamps than others.

Consider a node $v$ that is deleted in round $i$ as in item~(4) (or $v$ is active and marked $\false$ as in item~(5)), and a set $Z=\{z_1,\dots,z_k\}$ of its neighbors that become dirty by satisfying the condition in item~(4) (or item~(5)) above, ordered by their timestamps $(i,v,z_j)$ for $1\leq j\leq k$, as induced by this topology change. For each $1\leq j\leq k$, if a node $z_j$ becomes active due to this timestamp, then by item~(6), starting from round~$i$ none of its neighbors change their label to $\true$.
Consider the last round $i'$ before round $i$ in which the label of $z_j$ is $\true$ ($i'$ exists since this condition occurs initially when the graph is empty).
We claim that the topology change whose associated active node changed the label of $z_j$ to $\false$ in round $i'+1$, is either an insertion of an edge $\{z_j,u\}$ that satisfies the condition of item~(2) with $ID(z_j) < ID(u)$, or an insertion of the node $z_j$ which connects it to at least one node whose label is $\true$. The reason for this is that these are the only topology changes which cause $z_j$ to be assigned the label $\false$.

Finally, notice that these topology changes both induce only a single dirty node (thus a single active node and a single epoch), and therefore we can blame $z_j$ becoming active on the corresponding topology change.
This is an injective mapping, as any other node cannot blame these changes (they are changes that made $z_j$ dirty), and $z_j$ itself may become active again in the future due to satisfying the condition in item~(4) (or item~(5)) above only if its label is changed to $\false$ again in between.

In other words, this blaming argument implies $\epochs(i)\leq 2\cdot\changes(i)$ here as well, completing the proof.
\end{proof}

\section{Discussion}
\label{section:discussion}
This paper gives dynamic distributed algorithms for various fundamental tasks that have a constant amortized round complexity despite working in a very harsh environment.
We believe that fixing some LCL tasks with radius $r>1$ can be handled in a similar manner, but we leave this for future work.

A particular open question is whether one can improve the worst-case round complexity of our algorithm, perhaps to a logarithmic in $n$ complexity, as in~\cite{BambergerKM18}, without sacrificing its amortized round complexity.

~\\\textbf{Acknowledgments:} The authors are indebted to Yannic Maus for invaluable discussions which helped us pinpoint the exact definition of fixing in dynamic networks that we eventually use. 
We also thank Juho Hirvonen for comments about an earlier draft of this work, and Shay Solomon for useful discussions about his work in~\cite{AssadiOSS18,AssadiOSS19}. We thank Hagit Attiya and Michal Dory for their input about related work.

\bibliography{bibl}

\appendix


\section{Proof of Theorem~\ref{thm:lr}}

\begin{theorem-repeat}{thm:lr}\cite{BarYehudaE81,BarYehuda00}
	\TheoremLR
\end{theorem-repeat}

\label{sec:pf-lr}
\begin{proof}
	For every $v \in S_{\delta}$ it holds that 
	$w(v)\leq \sum_{e: v \in e}{\delta(e)}$. This gives:
	\begin{align*}
		\sum_{v \in S_{\delta}}{w(v)} \leq \sum_{v \in S_{\delta}}{\sum_{e: v \in e}{\delta(e)}} 
		\leq \sum_{v \in V}{\sum_{e: v \in e}{\delta(e)}}
		\leq  2 \sum_{e \in E}{\delta(e)}.
	\end{align*}
	It remains to prove that $\sum_{e \in E}{\delta(e)}\leq\OPT$.
	To this end, let $S_{\OPT}$ be a cover of minimal weight,
	and associate each edge $e\in E$ with its endpoint $v_e$ in $S_{\OPT}$ (choose an arbitrary endpoint if both are in $S_{\OPT}$).
	The weight $w(v)$ of each $v \in S_{\OPT}$ is at least $\sum_{e: v_e=v}{\delta(e)}$, because it is at least $\sum_{e: v \in e}{\delta(e)}$. Hence, $\OPT = \sum_{v \in S_{\OPT}}{w(v)} \geq \sum_{v \in S_{\OPT}}{\sum_{e: v_e=v}{\delta(e)}} = \sum_{e \in E}{\delta(e)}$,
	as desired.
\end{proof}

\section{Generalization: The family of locally fixable labelings}
\label{section:LFL general case}

To address the question of local fixability, we define a subclass of LCLs, which we call \emph{locally fixable labelings} (LFLs). Despite a deceiving first impression that such a definition may be straightforward, our definition turns out to be highly non-trivial, as we discuss in Appendix~\ref{sec:lfl}.
Our definition is entirely combinatorial, in the spirit of the definition of LCLs.

After defining LFLs, we present a simple, deterministic distributed template algorithm that rapidly fixes LFLs with labels made of small \emph{pieces}, a notion that will be made clear shortly. Specifically, we give an algorithm that uses only $O(1)$ communication rounds when amortized over the number of topology changes,
works even in the highly-dynamic setting where an \emph{unbounded number} of network links may appear or be dropped adversarially in every round of computation, and nodes can send messages of no more than $O(\log{n})$ bits.

All the fundamental LCL tasks discussed in our paper, i.e., maximal matching, $(degree+1)$-coloring, or 2-approximation for the minimum \emph{weighted} vertex cover (2-MWVC), can be described as LFLs with labels made of small pieces, which immediately implies the above fast fixing algorithm applies to these tasks, giving an alternative proof for our main theorems.

We further show that for some tasks, fast fixing is also possible with node insertions and deletions. Here, we consider the more severe case of \emph{abrupt} deletions, where  a deleted node does not have a chance to inform its neighbors of its upcoming departure from the system.

\subsection{LFLs and a fast dynamic fixing algorithm}
\label{sec:contribution}

\textbf{The challenge in defining LFLs.}
The first property that an LFL definition needs to have should pinpoint why the sinkless orientation example is unfixable in a fast manner.
Intuitively, in sinkless orientation, a local change may have a global impact on the labels of many and far away nodes, while an LFL needs to be fixed locally.
To rule out labelings with this undesired property, the LFL definition requires the existence of a \emph{local} fixing function. Such a function should potentially receive as input a star centered at a node $v$, along with its labeling which may be locally inconsistent, and produce fixed labels for $v$ and possibly its neighbors, such that local consistency holds for this star, and is maintained for all stars of these nodes if they were previously so.

However, this is still insufficient. Imagine the task of fixing a maximal matching when an edge is inserted to or deleted from the network graph. With a standard labeling that simply indicates which edge is matched, the endpoints of the changed edge can fix the labels of themselves and possibly their neighbors to produce a correct solution. Combinatorially, this is the fixing function that we refer to above. Yet, already here, a first subtlety arises: The endpoints cannot operate concurrently, because if a matching edge between them is deleted, they might try to match themselves to a neighbor that is common to both.

For a general task that one wishes to fix, this motivates simply requiring that the nodes that received indication of a topology change become \emph{active} for fixing at different times.
But this still does not address our aim of coping with an unlimited number of concurrent topology changes.
The reason is that by the time that a node becomes active and fixes the labels of its neighbors, we are no longer guaranteed that the old labels of this node and its neighbors correspond to an earlier locally consistent labeling, because of multiple topology changes that they may have undergone in the meanwhile. In other words, the old labels can be anything, so our fixing function is not promised to produce a locally consistent output.

This motivates an LFL definition that incorporates a two-step fixing algorithm. For the first step, at a high level, we define a function that \emph{prepares} a node for fixing by changing only its own label without the need for any communication. This produces a label that, when given as input to the fixing function, allows the fixing function to guarantee a locally consistent output.  Then, the second fixing step is the applying the fixing function itself, which may result in new labels for an active node and its neighbors.
For the maximal matching example, this intuitively can be seen as follows: when a node $v$ has its matched edge deleted from the graph, it first locally prepares itself by changing its label to indicate that it is now unmatched and ready for being matched if possible. Then, even before $v$ is able to become active for fixing, if one of its neighbors $u$ becomes active for fixing then $u$ can already get matched with $v$.

The above discussion gives the main intuitions for the two functions that underlie our LFL definition. To capture the properties that we need them to promise, we need several definitions of when the labels of neighbors are \emph{correct for their edge}, and when the view of a node is in essence locally consistent.
However, algorithmically, recall that we are limited to sending small messages in our setting, and hence may not be able to send entire labels for invoking the fixing function (this will become crucial when we show an LFL for fixing a 2-approximation for the minimum weight vertex cover). Hence, our labels need to be composed of small pieces such that the fixing function only requires a couple of pieces of every label. We call these pieces \emph{excerpts} and we thus need a stronger definition for views that replaces local consistency, because a node is unable to receive the entire label of a neighbor, but only a constant number of excerpts.

To summarize, an LFL will be defined as a tuple which, in addition to the input and output label sets $\Sigma$ and $\Gamma$, consists of a set that characterizes when the labels of neighbors are correct for their edge and a set that characterizes when the label of a node is ready for fixing (these two sets replace $\mathcal{C}$ as they induce local consistency), as well as two functions, one for preparing towards fixing and one for fixing. The formal full-fledged definition of LFLs is given in Appendix~\ref{sec:lfl}. When all inputs and all excerpts are small, we call this a \emph{bounded~LFL}.

~\\\noindent
\textbf{A fast dynamic fixing algorithm.}
In essence, because we might need a node to change the labels of its neighbors after a topology change, we make sure that concurrent fixing always happens for nodes that are not too close, and other nodes wait even if their labels are not yet correct. To this end, our method is to assign a time\-stamp to each change in the graph and fix a node that suffered from the change only if its timestamp is a local minimum in some neighborhood, thus avoiding conflicting concurrent fixes.
However, the restriction on the size of messages forbids unbounded timestamps.

To resolve this issue, we utilize ideas from the literature on shared memory algorithms, e.g.,~\cite{AttiyaDS89}, for deterministically hashing the timestamps into a small bounded domain so that the nodes can afford sending a hashed timestamp in a single small message, and we do so in a way that preserves the total order over timestamps. But this alone is still insufficient, because we need to cope with the uncontrolled number of topology changes, which may, for example, suddenly connect two nodes that were previously far enough to become active simultaneously, but can now interfere with each other. We carefully take care of such cases, where our timestamps have been \emph{cheated} by the topology changes.

Thus, our main algorithmic contribution lies in proving the following theorem, which holds in a model with an unbounded number of topology changes that may occur concurrently, and when only a logarithmic number of bits can be sent in a message.
This theorem is proven in Appendix~\ref{section:alg}.
\newcommand{\AlgTheorem}
{
	For every bounded LFL $\mathcal{L}$, there is a deterministic dynamic distributed fixing algorithm which handles edge insertions/deletions in $O(1)$ amortized rounds.
}
\begin{theorem}
	\label{theorem:alg}
	\AlgTheorem
\end{theorem}

While the definition of LFLs is unavoidably involved, our algorithm has the desired property of being simple, which we consider to be a benefit.

Theorem~\ref{theorem:alg} handles changes in edges, and a direct translation for handling node insertions/deletions (along with their edges) incurs an undesired blow-up in the complexity.
However, we prove that for some LFLs the same approach can also handle node insertions and deletions within the same complexity, because of additional properties that their LFLs satisfy.

\section{Locally-Fixable Labelings (LFLs)}
\label{sec:lfl}
The key observation of this paper is that we can pinpoint some LCLs that can be fixed fast.
The new concept of \emph{Locally-Fixable Labelings (LFLs)} is the key to our approach, and is a combinatorial definition in the spirit of LCLs.
The definition of LFL is quite involved, so we break it up into several parts. Here, we give a first glimpse into the structure of an LFL.
\begin{tcolorbox}
	\textbf{Locally-Fixable Labelings (LFLs), the main structure:} An LFL $\mathcal{L}$ is a tuple $(\Sigma, \Gamma, \mathcal{E}, \mathcal{P}, \prep, \fix)$, where $\Sigma$ is a set of input labels, $\Gamma$ is a set of output labels, and $\mathcal{E}, \mathcal{P}, \prep, \fix$ are defined below.
\end{tcolorbox}

We begin by imposing an inner structure for each output label in $\Gamma$, in a way which captures the part of the label that addresses the node itself, and the parts that address each possible edge incident on it. In what follows, a label typically refers to an output label, i.e., an element in $\Gamma$. Formally, we denote by $V=[n]$ the set of all possible nodes. A label $L_v$ of a node $v\in V$ is a vector indexed by the graph nodes, denoted
$L_v=(L_{v1},\ldots,L_{vv},\ldots,L_{vn})$.
The entry $L_{vv}$ is called the \emph{node excerpt}, and holds information about the node $v$.
Each entry of the form $L_{vu}$, for $u\neq v$, is called an \emph{edge excerpt}, and holds information about the pair $\{v,u\}$. Naturally, for non-neighboring nodes, the content of their edge excerpts is quite useless, and is therefore either chosen to be empty or to be some default value.

\subsection{Consistency}
\label{subsec:consistency}
We now focus on how to capture local consistency. Namely, we define the sets $\mathcal{E}$ and $\mathcal{P}$, which replace the set $\mathcal{C}$ in the definition of LCLs.

In essence, we want an LFL to be a type of LCL, and so we need that for each node, the label of itself and its neighbors determine whether the labeling is locally consistent. However, algorithmically, this requires that each node receives the labels of its neighbors, which in general may be too large to fit in a single message. Instead, we require that for each node $v$, local consistency of the labeling can be determined based only on its own label $L_v$, and the node excerpts of its neighbors, namely $\{L_{uu}\mid u\in N_v\}$. Algorithmically, when each node excerpt fits in a single message, this information can be obtained by $v$ in a single round.

We denote by $X$ the set of all possible excerpts of labels in $\Gamma$, i.e., $\Gamma\subseteq X^n$.
We denote by $\mathcal{N}_v$ the set of possible centered stars of node $v\in [n]$, and $\mathcal{N}=\bigcup_{v\in [n]}{\mathcal{N}_v}$.
An element of $\Gamma\times X^{d}$ for any value of $d\geq 0$ is called a \emph{view}. For each $v\in [n]$ and $N_v \in \mathcal{N}_v$ of size $d(v)+1$, we define below a set of views in $\Gamma\times X^{d(v)}$ which are called \emph{$N_v$-consistent} views. Intuitively, these are views whose excerpts correspond to a locally consistent labeling of $N_v$.
In a nutshell, $N_v$-consistent views are views for which all edge excerpts for edges touching $v$ comply with some correctness condition, and the label of $v$ itself follows another correctness rule. We derive the formal definition of $N_v$-consistent views through the two definitions: the set $\mathcal{E}$ of \emph{edge-correct tuples} and the set $\mathcal{P}$ of \emph{prepared labels} (i.e., prepared to be fixed).

~\\\textbf{Defining $\mathcal{P}$:} For each $v\in [n]$ and $N_v \in \mathcal{N}_v$, an LFL determines a set $\mathcal{P}$ of centered stars with labeled centers. Intuitively, these labels are prepared, in the sense that they can later be fixed if needed. We call the label of the center $v$ of each element in $\mathcal{P}$ an $N_v$-prepared label. Note that, algorithmically, no communication is needed for a node $v$ to check whether its label $L_v$ is $N_v$-prepared for its current neighborhood $N_v$, or to choose such a label.

~\\\textbf{Defining $\mathcal{E}$:} 
An LFL determines a set $\mathcal{E}$, whose elements are unordered pairs of nodes $\{v,u\}$ along with tuples $(L_{vv},L_{vu},L_{uv},L_{uu})$ in $X^4$, which we call edge-correct tuples. For each $\{v,u\}$-correct tuple, it must hold that $L_{vu}=L_{uv}$. We call this latter condition the \emph{reciprocity property}. Intuitively, an element in $\mathcal{E}$ corresponds to two labels that agree in their reciprocal edge excerpts and satisfy some required consistency.

~\\\textbf{$N_v$-consistency:}  For each $v\in [n]$ and $N_v \in \mathcal{N}_v$ of size $d(v)+1$, a view $(L_v,L_{u_1u_1},\dots,L_{u_{d(v)}u_{d(v)}})$ in $\Gamma\times X^{d(v)}$, for which $L_v$ is $N_v$-prepared and for every $u_i \in N_v$ the tuple $(L_{vv},L_{vu_i},L_{u_iv},L_{u_iu_i})$ is $\{v,u_i\}$-correct is called an \emph{$N_v$-consistent} view.
We can now consider a set $\mathcal{VC}$ (which stands for view consistency), which consists of all stars centered at each node $v$ along with a label for $v$, and a node excerpt and a $\{u,v\}$ edge excerpt for each neighbor of $v$, such that the obtained view is $N_v$-consistent. We do not include $\mathcal{VC}$ as a separate element in the definition of LFLs because this set is completely determined by the sets $\mathcal{P}$ and $\mathcal{E}$.

However, the set $\mathcal{VC}$ is useful for defining the analog to LFLs of an $\mathcal{L}$-legal labeling. For an LFL $\mathcal{L}$, a labeling $\lambda: V \rightarrow \Sigma \times \Gamma$ is called $\mathcal{L}$-legal for a graph $G=(V,E)$, if for every $v\in V$, there exists in $\mathcal{VC}$ an element $(H, s)$ with a label at the center and two excerpts at each other node, and there exists a mapping $\pi$ that maps the star centered at $v$ in $G$ into $(H,s)$ with $\pi(v)=s$, such that the label in $\lambda$ of $v$ is the same as the label of the node $s$ in $(H,s)$, and the respective excerpts of each neighbor $w$ of $v$ are the same as those of $\pi(w)$. If input labels exist, then they have to correspond in the mapping as well.

Crucially, if a view is $N_v$-consistent, then it is $N_v$-consistent regardless of the content of any other excerpts of the nodes $u_i \in N_v$ and any labels of other nodes. This implies that LFLs are a subclass of LCLs. To prove this, given an LFL we simply define a set $\mathcal{C}$ which contains each star $N_v$ centered at $v$, with node labels such that the view $(L_v,L_{u_1u_1},\dots,L_{u_du_d})$ is $N_v$-consistent, where the nodes $u_i$ are the neighbors of $v$. Clearly, if a labeling is $\mathcal{L}$-legal for the LFL, then it is also legal for the obtained LCL, which completes the argument. Notice that the other equivalence also holds: every LCL can be defined by replacing  $\mathcal{C}$ with $\mathcal{E}$ and $\mathcal{P}$, by plugging every label used when defining $\mathcal{C}$ into all excerpts of a new label which is then used for defining $\mathcal{E}$ and $\mathcal{P}$ accordingly. However, not every LCL is an LFL, because one cannot always find functions $\prep$ and $\fix$ as we require in what follows.

\subsection{Preparedness (defining $\prep$)}
\label{subsec:prepared}
Algorithmically, to fix consistency for a node $v$ after a topology change, we will need $v$ to fix also the labels of its neighbors. The node $v$ cannot guarantee consistency for a neighbor $u$ if it did not hold before the fix, but we still must make sure that a new label that is assigned by $v$ to a neighbor $u$ is always prepared, even if it is not yet consistent with respect to the labels of the neighbors of $u$ (which is information that $v$ does not have). That is, we need that the labels of all nodes are always prepared, even if a topology change occurs. This implies that even before any communication takes place for fixing a topology change, each node can prepare its label for fixing (make it prepared).

The combinatorial definition we use to capture the above is the following.
A function $\prep : [n]\times \mathcal{N}\times\mathcal{N}\times \Gamma \rightarrow \Gamma$ is a \emph{preparing function} if
for each node $v$,
if $N_v^{old},N_v^{new}\in\mathcal{N}_v$ and $L_v^{old}$ is $N_v^{old}$-prepared, then $L_v^{new}= \prep(v,N_v^{old},N_v^{new},L_v^{old})$ is $N_v^{new}$-prepared, and $L_{vu}^{old}=L_{vu}^{new}$ for every $u \not\in N_v^{old}\oplus N_v^{new}$.

Algorithmic intuition: The function $\prep$ allows node $v$ to change its own node excerpt $L_{vv}^{old}$, and any edge excerpt $L_{vu}^{old}$ of a node $u$ that is also affected by the topology change.
For such a node $u$, this excerpt has to change in a consistent manner, so that $u$ also changes the excerpt $L_{uv}^{old}$ of its old label in a way that satisfies the reciprocity property in the definition of an edge-correct tuple. For a node $u \notin N_v^{old}\oplus N_v^{new}$, the excerpt $L_{uv}$ is not allowed to be changed, as $u$ is unaware of the topology change and we want this function to be applied by $v$ without communication, while preserving the reciprocity property on all edges.

\subsection{Fixability (defining $\fix$)}
\label{subsec:fixability}
We are now ready for the most crucial definition, which is the one that shows that views are enough for \emph{fixing} consistency, given that they are of prepared labels. Algorithmically, when a node $v$ needs to fix its label for consistency, it is given its old label and the node excerpt $L_{uu}$ of every neighbor $u$ for a prepared $L_u$, and $v$ computes a new label for itself and possibly new node excerpts for its neighbors.
Crucially, although $v$ does not have information other than the old node excerpts, the new node excerpts must not compromise the preparedness of the labels from which they originate and they must not change edge-correctness for edges of neighboring nodes.
Hence, if before the fix of a node $v$ it holds that the only incorrect edge of its neighbor $u$ is its edge with $v$, then the new excerpts that $v$ assigns to $u$ must make it $N_u$-consistent for its neighborhood.

Combinatorially, we define the following. A function $\fix: [n]\times \mathcal{N}\times \Gamma \times X^{d} \rightarrow \Gamma\times X^{d}$ is called a \emph{fixing function} if whenever the following hold
\begin{itemize}
	\item $v \in [n]$, $N_v^{new}\in\mathcal{N}_v$, $d=|N_v^{new}|$ and $N_v^{new}=\{u_1,\dots, u_{d}\}$,
	\item $L_v^{old}$ is $N_v^{new}$-prepared,
	\item $G^{new}$ is a graph that is consistent with $N_v^{new}$, with a labeling for its nodes such that for each $u_i$, its label $L_{u_i}$ is $N_{u_i}$-prepared, where $N_{u_i}$ is the neighborhood of $u_i$ in $G^{new}$, and
	\item the labels $L_v^{old},L_{u_1},\dots,L_{u_d}$ satisfy the reciprocity property for every pair $\{v,u_i\}$,
\end{itemize}
then, when denoting
\[(L_v^{new}, \beta_1, \dots, \beta_d)=
\fix(v,N_v^{new},L_v^{old}, L_{u_1u_1},\dots, L_{u_du_d}),\]
it holds that
\begin{itemize}
	\item[(a)]{(preparedness)} $L_v^{new}$ is $N_v^{new}$-prepared, and
	\item[(b)]  for every $u_i$ denoting by $L'_{u_i}$ the label for which $L'_{u_iw}=L_{u_iw}$ for every $w\in[n]\setminus\{u_i,v\}$, $L'_{u_iu_i}=\beta_i$ and
	$L'_{u_iv}=L_{vu_i}^{new}$,
	it holds that for every $1\leq i\leq d$:
	\begin{enumerate}
		\item{(edge-correctness)} the tuple $(L_{vv}^{new},L_{vu_i}^{new},L'_{u_i v},L'_{u_iu_i})$ is $\{v,u_i\}$-correct,
		\item{(preparedness for neighbor)} $L'_{u_i}$ is $N_{u_i}$-prepared, and
		\item{(edge-correctness for neighbor)} for every $w \in N_{u_i}$, if
		$(L_{u_iu_i},L_{u_iw},L_{w u_i},L_{ww})$ is $\{u_i,w\}$-correct, then
		$(L'_{u_iu_i},L'_{u_iw},L_{w u_i},L_{ww})$ is
		$\{u_i,w\}$-correct, where $L_w$ is the label of $w$ in the given labeling for $G^{new}$.
	\end{enumerate}
\end{itemize}
Notice the crucial property that follows from the definition: The new view $(L_v^{new},L'_{u_1u_1},\dots,L'_{u_du_d})$ is $N_v^{new}$-consistent. This holds because item (a) in the definition is promised to hold, as well as item (b1) for all neighbors $u_i \in N_v^{new}$.

\subsection{The full LFL definition}
\label{subsec:lfl}
We can now finally define locally fixable labelings, as follows.
\begin{tcolorbox}
	\textbf{Locally-Fixable Labelings (LFLs):} An LFL $\mathcal{L}$ is a tuple $(\Sigma, \Gamma, \mathcal{E}, \mathcal{P}, \prep, \fix)$, where $\Sigma$ is a set of input labels, $\Gamma$ is a set of output labels, $\mathcal{E}$ is a set of edge-correct tuples, $\mathcal{P}$ is a set of centered stars with prepared labels at their centers, the function $\prep$ is a preparing function and the function $\fix$ is a fixing function.
\end{tcolorbox}

\section{An $O(1)$ amortized dynamic algorithm for edge insertions/deletions}
\label{section:alg}

Our main result is that an LFL can be fixed deterministically within a constant amortized number of rounds. For bounded LFLs, i.e., when the excerpt sizes are bounded by $O(\log{n})$ bits, our algorithm also works with a bandwidth of $O(\log{n})$ bits. Here, we present  our main algorithm, and in Appendix~\ref{sec:alg-node} we prove that the same algorithm also works for node insertions, given an additional property of the $\prep$ function of the LFL, which is satisfied in some of our examples.

\subsection{Edge insertions/deletions}
\label{sec:alg-edge}

\begin{theorem-repeat}{theorem:alg}
	\AlgTheorem
\end{theorem-repeat}

\begin{proof}
	First, we assume that all nodes start with initial labels that are globally consistent.
	
	\medskip
	\noindent\textbf{The setup:}
	We denote $\gamma=5$.
	Let $F_i$ be a set of edge changes (insertions/deletions) that occur in round $i \geq 0$ (for convenience, the first round is round 0). With each change in $F_i$, we associate two \emph{timestamps} such that a total order is induced over the timestamps as follows: for an edge $e=\{u,v\}$ in $F_i$, we associate the timestamp $ts=(i,u,v)$ with node $u$, and the timestamp $(i,v,u)$ with node $v$. Since $u$ and $v$ start round $i$ with an indication of $e$ being in $F_i$, both can deduce their timestamps at the beginning of round $i$. We say that a node $v$ is the \emph{owner} of the timestamps that are associated with it. In each round, a node only stores the largest timestamp that it owns, and omits the rest.
	
	Notice that timestamps are of unbounded size, which renders them impossible to fit in a single message. To overcome this issue we borrow a technique of ~\cite{AttiyaDS89}, and we invoke a deterministic hash function $H$ over the timestamps, which reduces their size to $O(\log{n})$ bits, while retaining the total order over timestamps. The reason we can do this is that not every two timestamps can exist in the system concurrently. To this end, we define $h(i)=i \mod 3\gamma n$ and
	$H(ts)=(h(i), u,v)$ for a timestamp $ts=(i,u,v)$, and we define an order $\prec_H$ over hashed timestamps as the lexicographic order of the 3-tuple, induced by the following order $\prec_h$ over values of $h$. We say that $h(i)\prec_h h(i')$ if and only if one of the following holds:
	\begin{itemize}
		\item $0 \leq h(i) < h(i') \leq 2\gamma n$, or
		\item $\gamma n \leq h(i) < h(i') \leq 3\gamma n$, or
		\item $2\gamma n \leq h(i) < 3\gamma n$ and $0 \leq h(i') < \gamma n$.
	\end{itemize}
	If two timestamps $ts=(i,v,u)$, $ts'=(i',v',u')$ are stored in two nodes $v,v'$ at two times $i,i'$, respectively, it holds that $ts<ts'$ if and only if  $H(ts) \prec_H H(ts')$.
	The reason that this holds despite the wrap-around of hashed timestamps in the third bullet above, is the following property that we will later prove: for every two such timestamps, it holds that $i'-i \leq \gamma n$. This implies $h(i) \prec_h h(i')$ whenever $i < i'$ despite the bounded range of the function $h$.
	
	In the algorithm, the nodes chop up time into \emph{epochs}, each consists of $\gamma$ consecutive rounds, in a non overlapping manner. That is, epoch $j$ consists of rounds $i=\gamma j,\dots,\gamma (j+1)-1$.
	
	\medskip\noindent\textbf{The algorithm:}
	For every epoch $j \geq 0$, we consider a set $D_j \subseteq V$ of \emph{dirty} nodes at the beginning of each epoch, where initially no node is dirty ($D_0=\emptyset$).
	Some nodes in $D_j$ may become \emph{clean} by the end of the epoch, so at the end of the epoch the set of dirty nodes is denoted by $D'_j$, and it holds that $D'_j \subseteq D_j$. At the beginning of epoch $j+1$, all nodes that receive any indication of an edge in $F_i$ in the previous epoch are added to the set of dirty nodes, i.e., $D_{j+1} = D'_{j} \cup I_j$, where $I_j$ is the set of nodes that start round $i$ with any indication about $F_i$, for any $\gamma j \leq i \leq \gamma (j+1)-1$.
	
	Intuitively, the algorithm changes the labels so that \emph{the labels at the \textbf{end} of the epoch are consistent with respect to the topology that was at the \textbf{beginning} of the epoch}, unless they are labels of dirty nodes or of neighbors of dirty nodes.
	
	The algorithm works as follows. In epoch $j=0$, the nodes do not send any messages, but some of them enter $I_0$ (if they receive indications of edges in $F_i$, for $0 \leq i \leq \gamma-1$).

	Denote by $N_{v}^i$ the neighborhood of $v$ in round $i$, denote by $L^i$ the labeling at the beginning of round $i$, before the communication takes place, and denote by $\hat{L}^i_v$ the labeling at the end of the round. Unless stated otherwise, the node $v$ sets $\hat{L}^i_v \leftarrow L^i_v$ and $L^{i+1}_v \leftarrow \hat{L}^i_v$.
	Now, consider an epoch $j>0$. On round $\gamma j$ every node $v\in D_j$ locally applies
	$L_v^{\gamma j}\leftarrow \prep(v,N^{\gamma (j-1)}_v,N^{\gamma j}_v,\hat{L}_v^{\gamma j-1})$,
	where $\hat{L}_v^{\gamma j-1}$ is the label that $v$ has at the \emph{end} of round $\gamma j-1=\gamma (j-1)+4$, which, as we describe below, may be different from its label $L_v^{\gamma j-1}$ at the beginning of the round.\footnote{We stress that one can describe our algorithm with labels that can only change at the beginning of a round, but we find the exposition clearer this way.}
	Then, the node $v$ sends $L_{vv}^{\gamma j}$ to its neighbors. These are the labels for the graph $G_{\gamma j}$ which the fixing addresses, and this is how the algorithm leverages the preparing function $\prep$ --- by bringing all the labels to the common ground of being prepared for the neighborhoods in the same graph, $G_{\gamma j}$.
	
	On rounds $\gamma j+1$ to $\gamma j+3$ the nodes propagate the hashed timestamps owned by dirty nodes. That is, on round $\gamma j+1$, each node in $D_j$ broadcasts its hashed timestamp, and on the following two rounds all nodes broadcast the smallest hashed timestamp that they see (with respect to the order $\prec_H$). Every node $v$ in $D_j$ which does not receive a hashed timestamp that is smaller than its own becomes \emph{active}.
	
	On the last round of the epoch, $\gamma j+4$, every active node $v$ that has neighborhood $N^{\gamma j}_v=\{u_1,\ldots,u_d\}$ at round $\gamma j$ computes
	$(\ell_{v},\beta_1,\dots,\beta_d)\leftarrow\fix(v,N^{\gamma j}_v,L^{\gamma j}_v,L_{u_1u_1}^{\gamma j},\dots,L_{u_du_d}^{\gamma j})$.
	Notice that $v$ has the require information to compute the above, even if additional topology changes occur during the rounds in which timestamps are propagated. Yet, we need to cope with the fact that topology changes may occur also throughout the current epoch and, for example, make active nodes suddenly become too close. For this, we denote by $T_j\subseteq I_j$ the set of \emph{tainted} nodes who received an indication of a topological change for at least one of their edges during the epoch $j$.
	
	Now, only if $v\notin T_j$ is an active node, it sends each neighbor $u$ the values $\ell_{vu}$ (the corresponding excerpt of $\ell_v$) and $\beta_u$ and becomes clean. Otherwise, an active node $v$ that is tainted (i.e., is in $T_j$) aborts and remains dirty for the next epoch.
	Of course, if two nodes $u$ and $v$ are neighbors at the beginning of an epoch but not when $v$ sends the computed excerpts, then $u$ does not receive this information.
	
	Finally, every active node $v\notin T_j$ updates $\hat{L}^{\gamma j+4}_v\leftarrow \ell_v$ and each neighbor $u$ updates $\hat{L}^{\gamma j +4}_u$ by setting $\hat{L}_{uu}^{\gamma j+4}\leftarrow\beta_u$ and $\hat{L}_{uv}^{\gamma j+4}\leftarrow\ell_{vu}$,
	leaving other excerpts unchanged. At the end of round $\gamma j+4=\gamma (j+1)-1$, node $v$ becomes \emph{inactive} (even if it aborts) and is not included in $D'_j$, i.e., we initialize $D'_j=D_j\setminus \{v\mid v \text{ is active in epoch } j\}$ at the beginning of epoch $j$. Notice that if $v$ is active but aborts then it will be in $I_j$ and hence in $D_{j+1}$.
	
	\noindent\textbf{Analysis:} In Appendix~\ref{app:alg}, we prove that the amortized round complexity is $O(1)$. In addition, we prove that the following invariant holds at the end of round $i=\gamma j+4=\gamma(j+1)-1$:
	\begin{enumerate}
		\item For every node $v$, its label $\hat{L}^{\gamma j+4}_v$ is $N^{\gamma j}_v$-prepared;
		\item For every two nodes $u,v$ that are clean at the end of the epoch and for which $\{u,v\} \in G_{\gamma j}$, it holds that the tuple $(\hat{L}_{uu}^{\gamma j+4},\hat{L}_{uv}^{\gamma j+4},\hat{L}_{vu}^{\gamma j+4}, \hat{L}_{vv}^{\gamma j+4})$ is $\{u,v\}$-correct.
	\end{enumerate}
	Since the invariants hold, we conclude that whenever $D_j=\emptyset$, it holds that for each edge $e$, the tuple associated with $e$ is $e$-correct, and all views are consistent for their respective neighborhoods in $G_{\gamma j}$, which by the correctness property implies that the labeling is in $\mathcal{L}$-legal. Further, what the invariants imply is that some correctness holds even for intermediate rounds: at the end of every epoch $j$, every tuple associated with an edge $e$ that touches two clean endpoints is $e$-correct, and the entire subgraph induced by nodes that are clean and have all of their neighborhood clean has locally consistent labelings.
\end{proof}

\section{Analysis of the Algorithm of Appendix~\ref{section:alg}}
\label{app:alg}
We provide here the missing details of the analysis of the algorithm (proof of Theorem~\ref{theorem:alg}).

~\\\textbf{Round complexity:}
We now prove that the algorithm has an amortized round complexity of $O(1)$. To this end, we show that any epoch $j$ in which messages are sent can be \emph{blamed} on a different timestamp $ts$ and that the node $v$ that owns $ts$ is either clean for the next epoch ($v \not\in D_{j+1}$) or is dirty because of a (new) change that occurs in one of its adjacent edges during epoch $j$ ($v\in I_j$).

First, we claim that for every two timestamps $ts=(i,v,u)$ and $ts'=(i',v',u')$ such that $ts<ts'$, that are simultaneously owned by nodes at a given time, it holds that $i'-i \leq \gamma n$. Assume otherwise, and consider the first time when this condition is violated for some timestamps $ts<ts'$. This means that the owner $v$ of $ts$ does not become active for more than $n$ epochs. Since in each epoch at least one timestamp is handled, $v$ not becoming active for more than $n$ epochs can only happen if at round $i$ there were more than $n$ timestamps which were then not yet handled, stored in various nodes. But there are at most $n$ nodes and each one stores at most one timestamps so the above is impossible. Since $i'-i \leq \gamma n$, we have that $H(ts) \prec_H H(ts')$, because $h(i) \prec_h h(i')$, as argued earlier.

Since the hashed timestamps are totally ordered by $\prec_H$, we have that in each epoch $j$ there is at least one dirty node $v$ that becomes active, namely the one with the minimal timestamp. The node $v$ is not in $D'_j$ and hence either is not in $D_{j+1}$ or is in $I_j$, as claimed. Since every topology change results in two timestamps, we have that the number of rounds required by the algorithm is $2\gamma=O(1)$, amortized over all changes.

~\\\textbf{Correctness:}
For correctness we claim the following invariant holds at the end of round $i=\gamma j+4=\gamma(j+1)-1$:

\begin{enumerate}
	\item For every node $v$, its label $\hat{L}^{\gamma j+4}_v$ is $N^{\gamma j}_v$-prepared;
	\item For every two nodes $u,v$ that are clean at the end of the epoch and for which $\{u,v\} \in G_{\gamma j}$, it holds that the tuple $(\hat{L}_{uu}^{\gamma j+4},\hat{L}_{uv}^{\gamma j+4},\hat{L}_{vu}^{\gamma j+4}, \hat{L}_{vv}^{\gamma j+4})$ is $\{u,v\}$-correct.
\end{enumerate}
We prove the above by induction on the epochs. Clearly the base case holds trivially as during the first epoch the labels do not change and we assume that the nodes start with an $\mathcal{L}$-legal labeling for the initial graph $G_0$. Now, assume the above invariants hold for epoch $j-1$. We analyze what happens for each item.
\begin{enumerate}
	\item We claim that at the end of the epoch, every node $v$ has a $N^{\gamma j}_v$-prepared label $\hat{L}^{\gamma j+4}_v$. By the induction hypothesis, $\hat{L}^{\gamma (j-1)+4}_v$ is a $N^{\gamma (j-1)}_v$-prepared labeling for $v$. Since we apply $\prep$ with $L_v^{old}=\hat{L}^{\gamma (j-1)+4}_v$ and $N^{old}=N^{\gamma (j-1)}_v$, preparedness holds by definition of $\prep$ for $L^{\gamma j}_v$.
	Now, if the label of $v$ does not change further during the epoch, then the invariant holds for it.
	
	If the label of $v$ changes because $v$ is an active node that does not abort, then it applies $\fix$ in order to obtain $\hat{L}^{\gamma j+ 4}_v$, which is $N^{\gamma j}_v$-prepared by item~(a) in the definition of $\fix$.
	
	Otherwise, we claim that the label of $v$ can only be changed by a single one of its neighbors. This is because if two of its neighbors, $u_1,u_2$ are active, then one of its edges to them must be inserted during the epoch because we propagate the timestamp to distance 3, which makes its endpoint abort. Since the label of $v$ is changed by a single neighbor that executes $\fix$, the item~(b2) in the definition of $\fix$ guarantees that the new label for $v$ is prepared with respect to its neighborhood in the respective graph, which is $G_{\gamma j}$.
	
	\item For every two nodes $u,v$ that are clean at the end of the epoch and for which $\{u,v\} \in G_{\gamma j}$, if their labels do not change during the epoch, then $\{u,v\}$-correctness of the respective tuple follows from the induction hypothesis.
	
	If only one of their labels changes, say that of $v$, then either $v$ applies $\fix$ or there is a (single) neighbor $w$ of $v$ which changes the label of $v$ by applying $\fix$. In the former case, by item~(b1) in the definition of $\fix$ it holds that the respective tuple of excerpts of the labels of $v$ and $u$ is $\{u,v\}$-correct. In the latter, $\{u,v\}$-correctness is given by item~(b3) in the definition of $\fix$.
	
	Finally, if both of their labels change, then either without loss of generality $v$ applies $\fix$ to both labels, in which case item~(b1) in the definition of $\fix$ promises that $\{u,v\}$-correctness holds, or $v$ and $u$ have a joint neighbor $w$ which applies $\fix$ and again $\{u,v\}$-correctness holds, by item~(b3). The crucial thing to notice here is that it cannot be the case that a node $w_v$ changes the label of $v$ and a different node $w_u$ changes the label of $u$ at the same time, because this implies that the distance between $w_v$ and $w_u$ is at most 3, in which case either at least one of them aborts due to an edge insertion, or the edge $\{u,v\}$ is inserted (maybe immediately after being deleted), but then $v$ and $u$ are not clean.
\end{enumerate}

Since the invariants hold, we conclude that whenever $D_j=\emptyset$, it holds that for each edge $e$, the tuple associated with $e$ is $e$-correct, and all views are consistent for their respective neighborhoods in $G_{\gamma j}$, which by the correctness property implies that the labeling is in $\mathcal{L}$-legal. Further, what the invariants imply is that some correctness holds even for intermediate rounds: at the end of every epoch $j$, every tuple associated with an edge $e$ that touches two clean endpoints is $e$-correct, and the entire subgraph induced by nodes that are clean and have all of their neighborhood clean has locally consistent labelings.

\subsection{Node insertions/deletions}
\label{sec:alg-node}

For node insertions and deletions, the proof of Theorem~\ref{theorem:alg} does not hold. Intuitively, this is because we need all neighbors of a changed node (inserted or deleted) to become dirty, which clearly increases the amortized complexity.

For the example of maximal matching, however, notice that when an edge is inserted, it suffices that \emph{only one} of its endpoints becomes dirty in the algorithm and fixes correctness for the tuple associated with that edge. Hence, if a node is inserted, it suffices that the inserted node becomes dirty, and we do not need all of its neighbors to become so. The property of the LFL for maximal matching which makes this possible is a property of the preparing function $\prep$, which makes it what we term as an \emph{insertion-closed preparing function}.

~\\\textbf{Insertion-closed preparing functions:} Formally, a preparing function $\prep$ is called \emph{insertion-closed} if for each $v\in[n]$, whenever $N_v^{old}\subseteq N_v^{new}$ it holds that $\prep(v,N_v^{old},N_v^{new},L_v^{old})=L_v^{old}$.

We claim that for every bounded LFL with an insertion-closed preparing function, our approach also handles node insertions within $O(1)$ amortized rounds.

\newcommand{\AlgTheoremNodes}
{
	For every bounded LFL  $\mathcal{L}$ with an insertion-closed preparing function, there is a deterministic dynamic distributed fixing algorithm which handles edge insertions/deletions and node insertions in $O(1)$ amortized rounds.
}
\begin{theorem}
	\label{theorem:alg-nodes}
	\AlgTheoremNodes
\end{theorem}

\begin{proof}
	When a node is inserted, it adopts an initial label that is $N_v$-prepared.
	The algorithm is almost exactly the same as the main algorithm, with the only modification being that a node $v$ becomes dirty also in the case when it is an inserted node.
	The analysis of $O(1)$ amortized number of rounds remains the same, and so we only argue that the two claimed correctness conditions hold at the end of round $\gamma j + 4 = \gamma (j+1) -1$:
	\begin{enumerate}
		\item For every node $v$, its label $\hat{L}^{\gamma j+4}_v$ is $N^{\gamma j}_v$-prepared;
		\item For every two nodes $u,v$ that are clean at the end of the epoch and for which $\{u,v\} \in G_{\gamma j}$, it holds that the tuple $(\hat{L}_{uu}^{\gamma j+4},\hat{L}_{uv}^{\gamma j+4},\hat{L}_{vu}^{\gamma j+4}, \hat{L}_{vv}^{\gamma j+4})$ is $\{u,v\}$-correct.
	\end{enumerate}
	The proof is again by induction and below we only indicate the modifications required compared with the proof of Theorem~\ref{theorem:alg}.
	\begin{enumerate}
		\item The proof for $N^{\gamma j}_v$-preparedness remains the same.
		\item For the proof of $\{u,v\}$-correctness to go through, we argue that $\{u,v\}$-correctness still holds after the insertion of a node $w$ that is a neighbor of at least one of $u$ and $v$, even before the dirty node $w$ becomes active. This is because the insertion-closed preparing function promises that the labels of $u$ and $v$ remain prepared after the insertion of $w$ without the need to change them, and hence the correctness of the respective tuple for the edge $\{u,v\}$ is maintained. \qedhere
	\end{enumerate}
\end{proof}

At this point, a natural question is why does our algorithm not handle also node deletions. The intuition for why this happens is as follows. By the definition of an insertion-closed preparing function, for an edge insertion it suffices that only one endpoint of the edge invokes $\fix$. This is why we can extended the algorithm to also apply to node insertions, as the inserted node is responsible for fixing the correctness of all of its edges. However, for an edge deletion, it may be the case that both endpoints need to invoke $\fix$ (imagine the maximal matching example, when a matching edge is removed from the graph, and so both endpoints need to match themselves with new neighbors). Hence, while for an edge deletion this only incurs a factor of 2 in the amortized complexity, when a node is deleted it may be that all of its former neighbors need to invoke $\fix$, which is too costly. 

Despite that, one may observe that a notion of \emph{deletion-closed preparing functions} could be defined in a similar manner. However, we suspect that such a notion has a much more limited applicability. Formally, we would like that for each $v\in[n]$, whenever $N_v^{old}\supseteq N_v^{new}$ it holds that $\prep(v,N_v^{old},N_v^{new},L_v^{old})=L_v^{old}$. Note that as with the intuition about maximal matchings, in LFLs whose preparing functions are either insertion-closed or deletion-closed, it is sufficient that only one endpoint of an edge becomes dirty to fix the labeling and ensure edge-correctness.  We emphasize that crucially in such a scenario with deletion-closed preparing functions there is no additional edge-correctness constraint to fix, and therefore node deletions may be handled without communication! Since we observe this property only in one example below, and as it seems to us more limiting, we chose not to study it further.

Nevertheless, for some problems, we can show that we can also handle node deletions efficiently. This happens when we can bound the number of neighbors that need to invoke $\fix$ after a node deletion (imagine the maximal matching example, where a deleted node could have been matched to at most a single neighbor, and so at most one former neighbor really needs to invoke $\fix$). In the paper, we address node deletions where applicable.

\end{document}